%% file: StochFundamentalLemmaMPC.tex
\pdfoutput=1 
\documentclass[letterpaper, 10 pt, conference]{ieeeconf}  

\IEEEoverridecommandlockouts                              

\usepackage{amsmath} 
\usepackage{amssymb}  
\input{package_macro.tex}

\title{On a Stochastic Fundamental Lemma and \\ Its Use for Data-Driven Optimal Control}

\author{Guanru Pan$^{\dagger}$, Ruchuan Ou$^{\dagger}$ and Timm Faulwasser$^{\star}$
\thanks{$^{\dagger}$: Equally contributing first authors. $^{\star}$: Corresponding author.}%
\thanks{Guanru Pan, Ruchuan Ou and Timm Faulwasser are with Institute for Energy Systems, Energy Efficiency and Energy Economics, TU Dortmund University, Dortmund,  Germany
        {\tt\small $\{$guanru.pan,ruchuan.ou$\}$@tu-dortmund.de and timm.faulwasser@ieee.org} }
}

\begin{document}

\maketitle
\thispagestyle{empty}
\pagestyle{empty}

\begin{abstract}
Data-driven control based on the fundamental lemma by Willems et al. is frequently considered for deterministic LTI systems subject to measurement noise. However, besides measurement noise, stochastic disturbances might also directly affect the dynamics. In this paper, we leverage Polynomial Chaos Expansions (PCE) to extend the deterministic fundamental lemma towards stochastic systems. This extension allows to predict future statistical distributions of the inputs and outputs for stochastic LTI systems in data-driven fashion, i.e., based on the knowledge of previously recorded input-output-disturbance data and of the disturbance distribution we perform  data-driven uncertainty propagation. 
Finally, we analyze data-driven stochastic optimal control problems and we propose a conceptual framework for data-driven stochastic predictive control.  Numerical examples illustrate the efficacy of the proposed concepts.
\end{abstract}


\input{Resource/Sec1_Introduction.tex}

\input{Resource/Sec2_ProblemStatement.tex}

\input{Resource/Sec3_DataDrivenStocha.tex}

\input{Resource/Sec4_StochasticOCPs.tex}

\input{Resource/Sec5_Numericalexample.tex}

\input{Resource/Sec6_Conclusion.tex}

\input{Resource/Sec7_Appendix.tex}

\bibliographystyle{IEEEtran}
\bibliography{IEEEabrv}

\end{document}

%% file: package_macro.tex
\let\labelindent\relax
\usepackage{enumitem} 

\usepackage{amsfonts}
\usepackage{amsmath}
\usepackage{mathrsfs}
\usepackage{comment}
\usepackage[normalem]{ulem}
\usepackage{amssymb,graphicx}

\usepackage{algorithm}
\usepackage{algorithmicx}  
\usepackage{algpseudocode}
\usepackage{booktabs}
\usepackage{multirow}
\usepackage{adjustbox}
\usepackage{siunitx}
\renewcommand{\algorithmicrequire}{\textbf{Input: }}

\newcommand{\End}{\hfill $\square$}
\DeclareMathOperator*{\argmin}{argmin}

\newcommand{\mbb}[1]{\mathbb{#1 }}
\newcommand{\mbf}[1]{\mathbf{#1}} 
\newcommand{\mcl}[1]{\mathcal{#1}}

\newcommand{\pce}[1]{\mathsf{#1}}
\newcommand{\pcecoe}[2]{\mathsf{#1}^{#2}}
\newcommand{\relx}{(\omega)}

\newcommand{\trar}[2]{\mathbf{#1}_{[0,{#2}]}}
\newcommand{\set}[1]{\mathbb{#1}}    

\newcommand{\inst}[1]{_{#1}}

\newcommand{\splx}[1]{\mcl{L}^2(\Omega, \mathcal{F}, \mu; \mathbb{R}^{#1})}
\newcommand{\spl}{\mcl{L}^2(\Omega, \mathcal{F}, \mu; \mathbb{R})}

\newcommand{\diff}{\mathop{}\!\mathrm{d}}
\newcommand{\Tini}{T_{\text{ini}}}
\newcommand{\mean}{\mbb{E}}
\newcommand{\var}{\mbb{V}}
\newcommand{\prob}{\mbb{P}}
\newcommand{\Hankel}{\mcl{H}}

\newcommand{\rinv}{\dagger}

\newcommand{\ini}{_{\text{ini}}} 

\newcommand{\dimy}{n_y}
\newcommand{\dimx}{n_x}
\newcommand{\dimu}{n_u}
\newcommand{\dimw}{n_w}
\newcommand{\dimz}{n_z}

\newcommand{\I}{\mathbb{I}}
\newcommand{\N}{\mathbb{N}}
\newcommand{\R}{\mathbb{R}}

\newtheorem{problem}{Problem}
\newtheorem{theorem}{Theorem}
\newtheorem{corollary}{Corollary}
\newtheorem{lemma}{Lemma}
\newtheorem{proposition}{Proposition}
\newtheorem{definition}{Definition}
\newtheorem{remark}{Remark}
\newtheorem{example}{Example}

\newtheorem{assumption}{Assumption}

\newcommand{\Xini}{ X_{\text{ini}}}
\newcommand{\xini}{ x_{\text{ini}}}
\newcommand{\xinipce}{ \pce{x}_{\text{ini}}}

%% file: Resource/Sec1_Introduction.tex
\section{INTRODUCTION}
Recently,  data-driven system representations based on the fundamental lemma by Willems et al. \cite{Willems2005} are subject to renewed and increasing research interest. The pivotal insight of the lemma is that the trajectories of any controllable Linear Time Invariant (LTI) system can be described without explicit identification of a state-space model. Specifically, provided persistency of excitation holds, the system trajectories are contained in the column space of a Hankel matrix constructed from recorded trajectories of input and output data.
In absence of process disturbances and measurement noise, this data-driven system representation is exact. Beyond the deterministic controllable LTI setting, there are recent variants of the lemma, e.g., extensions to nonlinear systems \cite{Alsalti21,lian21k}, to linear parameter-varying systems \cite{Verhoek21}, and to linear network systems \cite{Allibhoy20}. Other extensions include uncontrollable systems \cite{Mishra20,Yu2021}, and input affine systems \cite{Berberich21,Martinelli2022}. For recent overviews we refer to \cite{DePersis19, Markovsky21r}.

Data-driven control design and system analysis with not necessarily persistently exciting input data has been investigated in \cite{VanWaarde20}. The exploitation for predictive control 
has been popularized by \cite{Coulson2019}, while an earlier attempt can be found in \cite{Yang15}. For stability analysis of data-driven predictive control see~\cite{Berberich20}, while applications are discussed in \cite{Huang19,Lian21,Berberich21at,Carlet20,Bilgic22}.

Beyond the LTI setting, \cite{Coulson2019} proposes a heuristic approach to deal with measurement noise and mild system nonlinearities by introducing slack variables and regularization in the objective function. There is also a line of research focusing on the robustness with  respect to measurement noise and/or process disturbance: while \cite{Berberich20r,De21l,Van20n} consider the design of robust state feedback controllers to deal with process disturbance, \cite{Yin20} uses maximum likelihood to obtain an optimal Hankel representation, and  \cite{Coulson21a} views the measurement noise entering the Hankel matrix as a problem of distributional robustness.

The ultimate journal paper of Jan C. Willems~\cite{Willems12} as well as~\cite{Baggio17} provide a starting point for behavioral concepts for open stochastic systems. In a follow-up to the present paper we provide further results in this direction in~\cite{tudo:faulwasser22f}. However, to the best of the authors' knowledge, so far there appears to be no stochastic variant of the fundamental lemma. Moreover, intrusive uncertainty propagation and quantification---i.e., not relying on sampling or scenarios---and consequently the data-driven forward propagation of stochastic uncertainties through LTI dynamics represented by Hankel matrices are also open.

In the context of stochastic optimal control, and uncertainty quantification in general, Polynomial Chaos Expansions (PCEs) are an established method that can be applied in Markovian and non-Markovian settings. Its core idea is based on the observation that under mild technical assumptions random variables can be regarded as $\mcl{L}^2$ functions in a probability space and hence they admit representations in appropriately chosen polynomial bases. 
We refer to \cite{sullivan15introduction} for a general introduction to PCE and to \cite{Mesbah16,Heirung18} for recent overviews on stochastic model predictive control. Early works, which have popularized PCE for systems and control, include \cite{Fagiano12,kim13wiener,paulson14fast,Mesbah14,Sergio17}. 
Moreover, PCE allows computing statistical moments  efficiently \cite{Lefebvre20}, it has been used to analyze the region of attraction of stochastic systems \cite{Ahbe20}, and it finds application in power systems  \cite{Till19}. 

In this paper, we link the data-driven system representation via the fundamental lemma with the PCE approach for uncertainty propagation of LTI systems subject to process disturbance. Our contributions are as follows: 1) we present a stochastic variant of the fundamental lemma which enables prediction and propagation of the statistical distributions of the inputs and outputs over finite horizons.  The key observation is that the PCE coefficients of a stochastic LTI system satisfy dynamics with the same  matrices as the original system.
2) we present mild conditions under which a stochastic Optimal Control Problem~(OCP) in random variables can be formulated equivalently  in a finite-dimensional data-driven fashion without explicit knowledge of the system matrices. This reformulation is built upon knowledge or estimation of  disturbance realization trajectories. Hence 3), we also propose a strategy to estimate  disturbance realizations from input-output data without explicit system knowledge. Finally 4), drawing upon simulation examples, we demonstrate the efficacy of the proposed approach for data-driven stochastic optimal control.  

The remainder of the paper is as follows: Section~\ref{sec:problem_statement} gives details about the considered setting and revisits data-driven system representations. After a brief introduction to PCE, Section~\ref{sec:data_representation} and Section~\ref{sec:stochastic_OCPs} present the main results, i.e., the data-driven representation of stochastic LTI systems and the data-driven reformulation of a stochastic optimal control problem both using PCE. Section~\ref{sec:simulation} considers two numerical examples; the paper ends with conclusions in Section~\ref{sec:conclusion}.

%% file: Resource/Sec2_ProblemStatement.tex
\subsubsection*{Notation}
Let $(\Omega,\mathcal F,\mu)$ be a probability space with sample space~$\Omega$,  $\sigma$-algebra~$\mathcal F$, and probability measure~$\mu$.
Similarly, $\splx{n_z}$  is the space of random variables of dimension $n_z$ which have finite expectation and covariance. Let $Z: \I_{[0,T-1]} \rightarrow \splx{n_z}$ be a sequence of vector-valued random variables from time instant 0 to $T-1$.
We denote by $\mean[Z]$, $\var[Z]$, and $z\doteq Z(\omega) : \I_{[0,T-1]} \rightarrow \R^{n_z}$ its mean, variance, and realizations, respectively. The vectorization of $z$, respectively, $Z$ is written as $\trar{z}{T} \doteq [z_0^\top,z_1^\top, \dots,z_{T-1}^\top]^\top \in \R^{n_z T}$ and $\trar{Z}{T-1}$.
Throughout the paper, we denote the identity matrix of size $n$ by $I_n$ and $\|x\|_Q \doteq \sqrt{\frac{1}{2} x^\top Q x}$. For any matrix $Q\in \mbb{R}^{n\times m}$ with columns $q^1, \dots, q^m$, the column-space is denoted by $\mathrm{colsp}(Q) \doteq \mathrm{span}\left(\{q^1,\dots, q^m\}\right)$.

\section{Problem Statement \& Preliminaries} \label{sec:problem_statement}
\subsection{Model-Based Stochastic Optimal Control }

We consider stochastic discrete-time LTI systems
\begin{subequations}\label{eq:RVdynamics}
	\begin{align}
X\inst{k+1} &= AX\inst{k} +BU\inst{k}+ EW\inst{k}, \quad 
X_0=\Xini\\
Y\inst{k} &= CX\inst{k} +DU\inst{k}, \label{eq:RVdynamics_y}
	\end{align}
\end{subequations}
with state $X\inst{k}\in \mcl L^2(\Omega, \mcl F_k,\mu;\R^{n_x})$, input $U\inst{k}\in \mcl L^2(\Omega, \mcl F_k,\mu;\R^{\dimu})$, output $Y\inst{k}\in \mcl L^2(\Omega, \mcl F_k,\mu;\R^{\dimy})$, and process disturbance $W\inst{k}\in \mcl L^2(\Omega, \mcl F,\mu;\R^{\dimw})$ for $k\in \N$. In the underlying filtered probability space $(\Omega, \mcl F, (\mcl F_k)_{k\in \N}, \mu)$, the $\sigma$-algebra $\mcl F$ contains all available historical information, or more precisely,
\begin{equation}\label{eq:filtration}
\mcl F_0 \subseteq \mcl F_1 \subseteq ...  \subseteq \mcl F.
\end{equation}
Let $(\mcl F_k)_{k\in \N}$ be the smallest filtration that the stochastic process $X$ is adapted to, i.e.,
$
\mcl F_k = \sigma(X\inst{i},i\leq k)$,
where $\sigma(X\inst{i},i\leq k)$ denotes the $\sigma$-algebra generated by $X\inst{i},i\leq k$. Likewise, the stochastic input  $U\inst{k}$ and output $Y\inst{k}$ are also modelled as stochastic processes that are adapted to the filtration $(\mcl F_k)_{k\in \N}$, that is, $U\inst k$ and $Y\inst k$ only depend on $X\inst{0}, X\inst{1},...,X\inst{k}$. Note that the influence of the process disturbances $W\inst i,i\leq k$ is implicitly handled via the state recursion.
For more details on filtrations we refer to \cite{fristedt13modern}. 

Throughout the paper, we consider that the process disturbances $W\inst{k}$, $k\in \N$  are either identically independently distributed ($i.i.d.$), independently distributed, or dependently distributed random variables. For the sake of brevity, we present the main results for the case of $W\inst{k}$, $k\in \N$  being i.i.d.. We comment on the extension and major differences for the independently and dependently distributed cases in Remark~\ref{rem:dependentWk} in Section~\ref{sec:PCE_OCP}. Henceforth, we assume  that  the underlying probability distributions of $W\inst{k}$, $k\in \N$ are known. Additionally, the distribution of the initial condition $\Xini$ is also supposed to be known.

Our analysis commences with the following OCP. 
\begin{problem}[Stochastic OCP]\label{Problem1}
Given the initial condition $X_0=\Xini$ and random variables $W\inst{k},~k\in \I_{[0,N-1]}$, we consider the following OCP with horizon $N \in \N^+$,
\begin{subequations} \label{eq:stochasticOCP}
\begin{align}
   \min_{
\substack{
   \trar{X}{N-1} \in \splx{N\dimx}\\
      \trar{U}{N-1} \in \splx{N\dimu}  \\
      \trar{Y}{N-1} \in \splx{N\dimy}}
    }
   & \sum_{k=0}^{N-1} \mean \big[ \|Y_{k}\|^2_Q +\|U_{k}\|^2_R\big]  \label{eq:stochasticOCP_obj}\\
\text{subject to } \quad &\text{for } k \in \set I_{[0,N-1]}, \hspace{1cm} \notag \\
X_{k+1}=   AX_{k}+BU_{k} &
+E W_{k},\label{eq:OCPdynamic}
\quad X_0  = \Xini,
\\
Y_{k}=   CX_{k}+DU_{k} &
,\label{eq:OCPdynamic_Y}
\\
\prob [U_{k} \in  \set U ]\geq 1 - \varepsilon_u,& \label{eq:chance_U}\\
\prob [Y_{k}  \in  \set Y ]\geq 1 - \varepsilon_y,& \label{eq:chance_Y}
 \end{align}
\end{subequations}
where $Q \succeq 0$ and $R \succ 0$ are symmetric.
\End
 \end{problem}

Here,  we consider chance constraints for the inputs~\eqref{eq:chance_U} and the outputs \eqref{eq:chance_Y}. The underlying sets $\set{U}\subseteq \set{R}^{\dimu}$ and $\set{Y}\subseteq \set{R}^{\dimy}$ are assumed to be closed. Moreover, $1-\varepsilon_u$ and $1-\varepsilon_y$ specify the probabilities with which the---joint in the output dimension but individual in time---chance constraints shall be satisfied.\footnote{ However, notice that, depending on the considered distributions of process disturbances, the simultaneous consideration of chance constraints for inputs \textit{and} outputs may jeopardize feasibility of the OCP.}

For all $\omega \in \Omega$, the realization of $W\inst{k}$ is written as $
w\inst{k} \doteq W\inst{k}(\omega)$. We denote the state, input, and output realizations as $x_k \doteq  X\inst{k}(\omega) $, $u_k \doteq  U\inst{k}(\omega) $, and $y_k \doteq  Y\inst{k}(\omega) $, respectively.
Henceforth, we suppose that the system matrices $A$, $B$, $C$, $D$, $E$ as well as the future disturbance realizations $w_i, i\geq k$ are unknown; while the input realizations $u_k$ and  the output realizations $y_k$ are assumed to be known/measured. Moreover, we assume that for any minimal state realization of (1) the pair $(A,[B~E])$  is controllable and the pair $(A,C)$ is observable. 

\subsection{Primer on Data-Driven System Representation}
Given a specific initial condition $\xini = \Xini(\omega)$ and a sequence of disturbance realizations $w\inst{k}$ for $k \in \N$, the stochastic system~\eqref{eq:RVdynamics} induces the realization dynamics
\begin{subequations}\label{eq:RelizationDynamics}
\begin{align}
	x\inst{k+1} &= Ax\inst{k} +Bu\inst{k}+ Ew\inst{k}, \quad  
	x\inst{0}=\xini,\\
y\inst{k} &= Cx\inst{k} +Du\inst{k}.
\end{align}
\end{subequations}
We remark that for fixed input disturbance sequences $u_k, w\inst{k}$ for $k \in \N$ and a specific initial condition $\xini$, the realization dynamics~\eqref{eq:RelizationDynamics} are \textit{deterministic}.
Hence, the input and output trajectories of this LTI system can be represented using data. 
\begin{definition}[Persistency of excitation \cite{Willems2005}] Let $T, t \in \set{N}^+$. A sequence of inputs $\trar{u}{T-1}$ is said to be persistently exciting of order $t$ if the Hankel matrix
\begin{equation*}
\Hankel_t(\trar{u}{T-1}) \doteq \begin{bmatrix}
u\inst 0  & u\inst 1 &\cdots& u\inst{T-t} \\
u\inst 1  & u\inst 2 &\cdots& u\inst{T-t+1} \\
\vdots& \vdots & \ddots & \vdots \\
u\inst{t-1}&u\inst{t}& \cdots  & u\inst{T-1} \\
\end{bmatrix}
\end{equation*}
is of full row rank.  \End
\end{definition}

Since \eqref{eq:RelizationDynamics} is driven by the inputs and the realizations of the process disturbance, the extension of the fundamental lemma by Willems et al. to the exogenous input data $\trar{(u,w)}{T-1}$ is immediate.
\begin{lemma}[Deterministic fundamental lemma]\label{Lem:fundamental} ~ \\Let $T$, $t \in \set{N}^+$. Consider a realization trajectory $\trar{(u,w,y)}{T-1}$
of \eqref{eq:RelizationDynamics}. Assume that the pair $(A, [B~E])$ is controllable.
If $\trar{(u,w)}{T-1}$ is  persistently exciting of order $\dimx+t$, then 
$\trar{(\tilde{u},\tilde{w},\tilde{y})}{t-1}$ 
is a realization trajectory of \eqref{eq:RelizationDynamics} if and only if there exists a $g\in \R^{T-t+1} $ such that
\begin{equation} \label{eq:realization_funda}
\Hankel_t(\trar{z}{T-1})g= \trar{\tilde{z}}{t-1}\, \text{ holds }  \forall\mbf{z}\in \{\mbf{u},\mbf{w},\mbf{y}\}.
\end{equation} 
\End
\end{lemma}
\begin{proof}
In \cite{Willems2005} , the fundamental lemma is originally given and proven in the behavioral framework. Its reformulation in terms of  state-space descriptions can be found in \cite[Lemma~2]{DePersis19}.  Applying \cite[Lemma~2]{DePersis19} and considering $(u,w)$ as  exogenous inputs to the system, the  assertion follows directly.  
\end{proof}
At this point, it is fair to ask  for how to obtain---or how to estimate---previous disturbance realizations $\trar{w}{T-1}$? We postpone our answer  to Section~\ref{sec:estimation}. However, even temporarily assuming exact measurements of 
$\trar{w}{T-1}$,  the \textit{future disturbance realizations} $\trar{\tilde{w}}{t-1}$ on the right-hand side of \eqref{eq:realization_funda} are not known. Even when given the future realizations of the inputs $\trar{\tilde{u}}{t-1}$ without the knowledge of $\trar{\tilde{w}}{t-1}$, 
one cannot compute the future realizations of the outputs $\trar{\tilde{y}}{t-1}$ via~\eqref{eq:realization_funda}. 
Indeed, the ambition of our subsequent discussions is twofold: (i) the development of data-driven methods to predict the future evolution of the distributions of inputs and outputs based on the knowledge of past realizations and of the distribution of the process disturbance,  and (ii) the reformulation of Problem~\ref{Problem1} in a computationally tractable data-driven form.

%% file: Resource/Sec3_DataDrivenStocha.tex
\section{Data-driven  Representations of Stochastic LTI Systems}\label{sec:data_representation}
\subsection{Basics of Polynomial Chaos Expansion}
Polynomial Chaos Expansion (PCE) enables to propagate uncertainties through system dynamics and thus it provides an alternative to describe future random variables rather than accessing their realizations. Its origins date back to Norbert Wiener~\cite{wiener38homogeneous}; for a general introduction to PCE  see \cite{sullivan15introduction}. 

The core idea of PCE is that an $\mathcal{L}^2$ random variable can be expressed in a suitable polynomial basis. To this end, we consider an orthogonal polynomial basis $\{\phi^j \relx\}_{j=0}^\infty$ which spans $\spl$, i.e.,
\[
	\langle \phi^i, \phi^j \rangle \doteq \int_{\Omega} \phi^i(\omega)\phi^j(\omega) \diff \mu(\omega) = \delta^{ij} \langle \phi^j\rangle^2,
\]
where $\delta^{ij}$ is the Kronecker delta and $\langle \phi^j\rangle^2 \doteq \langle \phi^j,\phi^j\rangle$. 

\begin{definition}[Polynomial chaos expansion]
The PCE of a real-valued random variable $Z\in \spl$ with respect to the basis  $\{\phi^j \relx\}_{j=0}^\infty$ is
\begin{equation*}\label{eq:PCE_def}
Z = \sum_{j=0}^{\infty}\pcecoe{z}{j} \phi^j  \quad\text{with}\quad \pcecoe{z}{j} = \frac{\langle Z, \phi^j \rangle}{\langle \phi^j\rangle^2 },
\end{equation*}
where $\pcecoe{z}{j} \in \R$ is called the $j$-th PCE coefficient. \End
\end{definition}
We remark that by applying PCE component-wise the $j$-th PCE coefficient of a vector-valued  random variable $Z\in\splx{\dimz}$ reads
\begin{equation*}
	\pce{z}^j = \begin{bmatrix} \pcecoe{z}{1,j} & \pcecoe{z}{2,j} & \cdots & \pcecoe{z}{\dimz,j} \end{bmatrix}^\top,
\end{equation*}
where $\pcecoe{z}{i,j}$ is the $j$-th PCE coefficient of component $Z^i$. 

In numerical implementations, the series have to be terminated after a finite number of terms which may lead to truncation errors.
For details on truncation errors and error propagation see~\cite{field04accuracy,muehlpfordt18comments}. Indeed, random variables that follow some widely used distributions admit exact finite-dimensional PCEs in suitable polynomial bases, e.g., for Gaussian random variables the Hermite polynomials are chosen. 
\begin{definition}[Exact PCE representation] \label{def:exact_pce}
		A random variable $Z\in \splx{\dimz}$ is said to admit an exact PCE with $L$ terms if 
		\[
	 Z- \sum_{j=0}^{L-1} \pcecoe{z}{j} \phi^j  =0. \tag*{\End}
		\]
\end{definition}
We refer to \cite{koekoek96askey, xiu02wiener} for details.
\begin{remark}[Appropriate bases for exact PCE] 
Given an $\mcl{L}^2$ random variable with known distribution, the key to construct an exact finite-dimensional PCE is the appropriate choice of basis functions. For some widely used distributions, the appropriate choice of polynomial bases is summarized in Table~\ref{tab:askey_scheme}. Notice that one uses specific random-variable  arguments $\xi \in \splx{n_\xi}$ for different polynomial basis, cf. Table~\ref{tab:askey_scheme}. On the other hand, for random variables not listed in Table~\ref{tab:askey_scheme}
a trivial (non-orthogonal)
 basis choice is $\phi^0 = 1$ and $\phi^1 = Z$ which implies the exact and finite PCE $\mathsf z^0 = 0, \mathsf z^1 = 1$. In the framework of $\mcl{L}^2$ random variables one could also employ a Gram-Schmidt process to construct appropriate orthogonal basis functions~\cite{Witteveen2006}. \vspace*{1.5mm} \End
\end{remark}

Given  $Z, \widetilde{Z}\in \splx{\dimz}$ admitting exact PCEs of $L$ terms, cf. Definition \ref{def:exact_pce}, 
the expectation $\mean[Z]\in \R^{\dimz}$ , the variance $\mbb{V}[Z]\in \R^{\dimz}$, and the covariance $\Sigma [Z,\widetilde{Z}]\in \R^{\dimz\times \dimz}$ can be obtained from the PCE coefficients as $\mean [Z] = \pce{z}^0$,
	\begin{equation} \label{eq:MomentsPCE}
		\mbb V [Z] = \sum_{j=1}^{L-1} \big(\pce{z}^j\big)^2 \langle \phi^j\rangle^2,~ \Sigma[Z,\widetilde{Z}] = \sum_{j=1}^{L-1} \pce{z}^j\tilde{\pce{z}}^{j\top}\langle \phi^j\rangle^2,
	\end{equation}
where $(\pce{z}^j)^2 \doteq \pce{z}^j \circ  \pce{z}^j $ refers to the Hadamard product.
We refer to~\cite{Lefebvre20} for a detailed discussion.
\begin{table}[t] 
\caption{Correspondence of random variables and underlying orthogonal polynomials.}
\label{tab:askey_scheme}
\centering
		\begin{adjustbox}{width=\columnwidth,center}
\begin{tabular}{ccc c}
	\toprule
	Distribution  & Support & Orthogonal basis $\{\phi^j\}_{j=0}^{\infty}$ & Argument $\xi \relx$ \\
	\midrule
	Gaussian & $(-\infty, \infty)$ & Hermite & $\mathcal N(0,1)$ \\
	Uniform  & $[a,b]$ & Legendre &$\mathcal U ([-1,1])$ \\
	Beta & $[a,b]$ & Jacobi & $\mathcal{B}(\alpha,\beta,[-1,1])$ \\
	Gamma & $(0,\infty)$ & Laguerre & $\Gamma(\alpha,\beta,(0,\infty))$ \\
	\bottomrule
\end{tabular}
		\end{adjustbox}
\end{table}

\subsection{Fundamental Lemma for Stochastic LTI Systems}
Replacing all random variables of \eqref{eq:RVdynamics} with their PCE expansions with respect to the basis $\{\phi^j\relx\}_{j=0}^\infty$ and performing a so-called Galerkin projection onto the basis functions $\phi^j\relx$---we refer to Appendix~A for details---, we obtain the dynamics of the PCE coefficients. 

For all $j \in \N \cup \{\infty\}$, with  given
$\xinipce^j$ and $\pce{w}^{j}_k$, $k \in \N$, the dynamics of the PCE coefficients read
\begin{subequations}\label{eq:PCEcoesDynamics}
\begin{align}
\pcecoe{x}{j}\inst{k+1} &= A\pcecoe{x}{j}\inst{k} +B\pcecoe{u}{j}\inst{k}+ E \pcecoe{w}{j}\inst{k},\quad
\pcecoe{x}{j}\inst{0} = \xinipce^j,
\\
\pcecoe{y}{j}\inst{k}& = C\pcecoe{x}{j}\inst{k} +D\pcecoe{u}{j}\inst{k}.
\end{align}
\end{subequations}

Notice that the PCE coefficients of the initial state and process disturbances are determined with respect to their known distributions, and thus the PCE coefficient dynamics~\eqref{eq:PCEcoesDynamics} are \textit{deterministic}.
Therefore, they admit the conceptual application of the usual LTI fundamental lemma. 
\begin{lemma}[Fundamental lemma for PCE coefficients]\label{lem:PCEfundamental} ~\\
Let $T, t \in \set{N}^+$. For $j \in \N \cup \{\infty\}$, consider a  stacked PCE coefficient trajectory $\trar{\pce{(u, w, y)}}{T-1}^j$  of~\eqref{eq:PCEcoesDynamics}. Suppose that the pair $(A,[B~E])$ is controllable. For all $j \in \N \cup \{\infty\}$, let $\trar{(\pce{u},\pce{w})}{T-1}^j$ be persistently exciting of order $\dimx +t$. Then $\trar{ ( \tilde{\pce u}, \tilde{\pce w}, \tilde{\pce y})}{t-1}^j$ 
is a PCE coefficient trajectory of \eqref{eq:PCEcoesDynamics} if and only if there exists $\pcecoe{g}{j}\in \R^{T-t+1} $ such that
\begin{equation} \label{eq:PCEfunda}
\Hankel_t(\trar{\pce{z}}{T-1}^j) \pcecoe{g}{j}= \trar{ \tilde{\pce z}}{t-1}^j \,  \text{ holds  } \forall\pce{z}\in \{\pce{u},\pce{w},\pce{y}\},
\end{equation} 
and all $  j \in \N \cup \{\infty\}$.
 \End
\end{lemma}
The proof follows from Lemma~\ref{Lem:fundamental} and is thus omitted.

Lemma \ref{lem:PCEfundamental} as such is straightforward, but it is not trivial to measure or estimate PCE coefficients of a stochastic LTI system. Hence the previous result is seemingly not very practical.  However, as we show below the structural similarity of the PCE coefficient dynamics~\eqref{eq:PCEcoesDynamics} and the original stochastic system~\eqref{eq:RVdynamics} enables further useful insights. 

Consider the stochastic LTI system \eqref{eq:RVdynamics} and the corresponding trajectories of random variables, PCE coefficients driven by \eqref{eq:PCEcoesDynamics}, and realizations generated by \eqref{eq:RelizationDynamics}, which are $\trar{(U,W,Y)}{T-1}$, $\trar{\pce{(u,w,y)}}{T-1}^j, \, j \in \N \cup \{\infty\}$, and $\trar{(u,w,y)}{T-1}$,  respectively. We have the following pivotal result.

\begin{lemma}[Column-space equivalence]\label{lem:colEqui} 
Consider the stochastic LTI system \eqref{eq:RVdynamics} and its $\mcl L^2(\Omega, \mcl F,\mu;\R^{\dimz})$, $\dimz \in\{\dimu, \dimw,\dimy\}$ random-variable trajectories $\trar{(U,W,Y)}{T-1}$.  Suppose that the pair $(A,[B~E])$ is controllable. 

Let the corresponding PCE coefficient trajectories $\trar{\pce{(u, w)}}{T-1}^j$, $j \in \N \cup \{\infty\}$ and the realizations  $\trar{(u,w)}{T-1}$ be persistently exciting of order $\dimx +t$.
\begin{itemize}
\item[1)] Then, for all $j \in \N \cup \{\infty\}$ and all $(\pce{z},\mbf{z}) \in \{(\pce{u},\mbf{u}),(\pce{w},\mbf{w}),(\pce{y},\mbf{y})\} $,
 it holds that
\begin{subequations}
\begin{equation} \label{eq:colEqui}
\mathrm{colsp}\big(\Hankel_t(\trar{\pce{z}}{T-1}^j) \big) = \mathrm{colsp}\big(\Hankel_t(\trar{z}{T-1} )\big).
\end{equation}
\item[2)] Moreover, for all $g \in \R^{T-t+1}$, there exists $G \in  \splx{T-t+1} $ such that
\begin{equation}\label{eq:colEquiRV}
\Hankel_t(\trar{Z}{T-1}) g = \Hankel_t(\trar{z}{T-1}) G 
\end{equation}
\end{subequations}
holds for all $(\mbf{Z},\mbf{z}) \in \{(\mbf{U},\mbf{u}),(\mbf{W},\mbf{w}),(\mbf{Y},\mbf{y})\} $.
\End
\end{itemize}
\end{lemma}
\begin{proof}
For the sake of readability, we omit the subscript $(\cdot)_{[0,T-1]}$ in the proof.
The proof of \eqref{eq:colEqui} in Part 1) follows directly from the  observation that the realization dynamics~\eqref{eq:RelizationDynamics} and the PCE coefficient dynamics~\eqref{eq:PCEcoesDynamics} share the same system matrices $(A,B,C,D,E)$.

\textit{Part 2):} Considering \eqref{eq:colEquiRV}, we have
\begin{align*}
 \Hankel_t(\mbf{Z}) g &= \Hankel_t\left(\sum_{j=0}^{\infty}  \pce{z}^j\phi^j\right) g\\&= \sum_{j=0}^{\infty}\Hankel_t( \pce{z}^j\phi^j)g= \sum_{j=0}^{\infty} \phi^j \Hankel_t( \pce{z}^j)g.
\end{align*}
Note that $\Hankel(\cdot)$ and the summation are both linear operations, therefore the second equality holds. Moreover,  the basis function $\phi^j(\xi)$ is a scalar polynomial of its random-variable argument $\xi \in \splx{n_\xi}$. Hence we have $\phi^j(\xi) \in \spl$. In case the different components of $\mbf{Z}$ require different bases, one relies on the union of the bases for each component.  
Then, using the column space equivalence \eqref{eq:colEqui}, for all $j \in  \N \cup \{\infty\}$ and any $g \in \R^{T-t+1}$, we can find $\pcecoe{g}{j}\in \R^{T-t+1}$, such that
$
\Hankel_t(\pce{z}^j) g = \Hankel_t(\mbf{z} )\pcecoe{g}{j}$.
This leads to
\begin{align*}
\Hankel_t(\mbf{Z}) g &= \sum_{j=0}^{\infty} \phi^j \Hankel_t( \pce{z}^j)g 
= \Hankel_t(\mbf{z} )\sum_{j=0}^{\infty}\phi^j \pcecoe{g}{j}.
\end{align*}
The assertion follows with $G \doteq \sum_{j=0}^{\infty}\pcecoe{g}{j} \phi^j $.
\end{proof}

	The next result relaxes Part 1) of Lemma~\ref{lem:colEqui} with respect to the persistency of excitation of the PCE coefficient trajectories. 
	\begin{corollary}[Column-space inclusion] \label{cor:inclusion}
	Consider the stochastic LTI system \eqref{eq:RVdynamics} and its $\mcl L^2(\Omega, \mcl F,\mu;\R^{\dimz})$,
	$\dimz \in\{ \dimu, \dimw,\dimy\}$ random-variable trajectories $\trar{(U,W,Y)}{T-1}$.	Let the corresponding realizations  $\trar{(u,w)}{T-1}$ be persistently exciting of order $\dimx +t$. Suppose that the pair $(A,[B~E])$ is controllable. 
	
Then, for all $j \in \N \cup \{\infty\}$ and all $(\pce{z},\mbf{z}) \in \{(\pce{u},\mbf{u}),(\pce{w},\mbf{w}),(\pce{y},\mbf{y})\} $,
		it holds that
			\begin{equation*} \label{eq:colInclu}
				\mathrm{colsp}\big(\Hankel_t(\trar{\pce{z}}{T-1}^j) \big) \subseteq \mathrm{colsp}\big(\Hankel_t(\trar{z}{T-1} )\big). \tag*{\End}
\vspace*{1mm}			\end{equation*}
\end{corollary}

	Based on the two results above, we obtain the following data-driven representation  of the
	PCE coefficient dynamics~\eqref{eq:PCEcoesDynamics}  that uses only the knowledge of realization data.
	
\begin{corollary}[PCE coefficients via realizations]\label{cor:mixed_funda}
	Suppose that the conditions of Corollary~\ref{cor:inclusion} hold.
If the realizations $\trar{(u,w)}{T-1}$ are persistently exciting of order $\dimx +t$, then, for all  $j \in \N \cup \{\infty\}$, $\trar{ ( \tilde{\pce u}, \tilde{\pce w},\tilde{\pce y})}{t-1}^j $ is an input-output-disturbance PCE coefficient trajectory of \eqref{eq:PCEcoesDynamics} 
if and only if there exists $\pcecoe{g}{j}\in \R^{T-t+1}$ such that 
\begin{equation} \label{eq:mixed_funda}
\Hankel_t(\trar{z}{T-1}) \pcecoe{g}{j}= \trar{ \tilde{\pce z}}{t-1}^j
\end{equation}
holds for all $(\mbf{z}, \tilde{\pce{z}})\in \{ (\mbf{u},\tilde{\pce{u}}), (\mbf{w},\tilde{\pce{w}}),(\mbf{y},\tilde{\pce{y}})\} $.
 \End
\end{corollary}

Observe that the core difference between \eqref{eq:realization_funda} and  \eqref{eq:mixed_funda}  is that in the latter  the PCE coefficients $\tilde{\pce{w}}$ of future process disturbances are given, i.e., they are known, while~\eqref{eq:realization_funda} requires knowledge of future disturbance realizations. Hence, \eqref{eq:mixed_funda} allows to predict the future output distributions of stochastic LTI systems based on knowledge of the distributions of future inputs. Moreover, one can lift the results to the corresponding $\mcl{L}^2$ probability space. 

\begin{lemma}[Stochastic fundamental lemma]\label{lem:RVfundamental} ~\\
Consider the stochastic LTI system \eqref{eq:RVdynamics} and its $\mcl L^2(\Omega, \mcl F,\mu;\R^{\dimz})$, 
$\dimz \in\{ \dimu, \dimw, \dimy\}$ trajectories of 
 random variables, the corresponding PCE coefficient trajectories from \eqref{eq:PCEcoesDynamics}, and the corresponding realization trajectories from \eqref{eq:RelizationDynamics}, which are $\trar{(U,W,Y)}{T-1}$, $\trar{\pce{(u,w,y)}}{T-1}^j, \, j \in \N \cup \{\infty\}$, and $\trar{(u,w,y)}{T-1}$, respectively. Suppose that the pair $(A,[B~E])$ is controllable. 
\begin{subequations}
\begin{itemize}
\item[1)]  Let $\trar{(u,w)}{T-1}$ be persistently exciting of order $\dimx +t$. Then
$ (\widetilde{\mbf{U}}, \widetilde{\mbf{W}},\widetilde{\mbf{Y}})_{[0,t-1]}$ is a trajectory of \eqref{eq:RVdynamics} if and only if there exists $G \in \splx{T-t+1} $ such that 
\begin{equation} \label{eq:RVfunda}
\Hankel_t(\trar{z}{T-1}) G=\widetilde{\mbf{Z}}_{[0,t-1]}
\end{equation} 
 holds for all $(\mbf{z}, \widetilde{\mbf{Z}})\in \{(\mbf{u},\widetilde{\mbf{U}}), (\mbf{w}, \widetilde{\mbf{W}}),(\mbf{y},\widetilde{\mbf{Y}})\} $. 
\item[2)]   
Let $\trar{(\mbf{U},\mbf{W})}{T-1}$ satisfy
\[
\trar{\mbf{Z}}{ T-1}= \sum_{j=0}^{L-1} \pcecoe{z}{j} \phi^j, \quad \mbf{Z} \in \{\mbf{U}, \mbf{W}\}
\]
with $L \in \mbb{N}^+$ and all PCE trajectories $\trar{(\pce{u},\pce{w})}{T-1}^j$ with $j \in \{0,\dots, {L-1}\}$ are persistently exciting of order $\dimx +t$. If there exists a $g\in \R^{T-t+1} $ such that
\begin{equation} \label{eq:RVfunda2}
\Hankel_t(\trar{Z}{T-1}) g= \widetilde{\mbf{Z}}_{[0,t-1]}
\end{equation} 
 holds for all $\mbf{Z} \in \{\mbf{U},\mbf{W},\mbf{Y}\} $, 
then $ (\widetilde{\mbf{U}}, \widetilde{\mbf{W}},\widetilde{\mbf{Y}} )_{[0,t-1]}$ is a trajectory of \eqref{eq:RVdynamics}. 
\End
\end{itemize}
\end{subequations}
\end{lemma}
\begin{proof}
\textit{Part 1):}	 We begin with showing that  $G \mapsto \tilde{\mbf{Z}}$. Consider $G \in \splx{T-t+1} $ which can be written as
$G = \sum_{j=0}^{\infty}\pcecoe{g}{j} \phi^j$.
The PCE coefficients $\pcecoe{g}{j} \in \R^{T-t+1}$ determine PCE coefficient trajectories of \eqref{eq:PCEcoesDynamics}, cf. Corollary~\ref{cor:mixed_funda}.
Multiplying \eqref{eq:mixed_funda} with the polynomial basis $\phi^j$ gives
\begin{equation*}
\Hankel_t(\trar{z}{T-1})\pcecoe{g}{j}\phi^j = \trar{ \tilde{\pce z}}{t-1}^j\phi^j, \quad \forall j \in \N \cup \{\infty\}.
\end{equation*}
We obtain 
\begin{equation*}
\Hankel_t(\trar{z}{T-1}) G = \Hankel_t(\trar{z}{T-1})\sum_{j=0}^\infty \pcecoe{g}{j}\phi^j =\sum_{j=0}^\infty  \trar{ \tilde{\pce z}}{t-1}^j\phi^j.
\end{equation*}
Hence $G$ determines a random variable trajectory  $\widetilde{\mbf{Z}}_{[0,t-1]} = \sum_{j=0}^\infty \trar{ \tilde{\pce z}}{t-1}^j\phi^j $ of \eqref{eq:RVdynamics}. 

Next we show $\widetilde{\mbf{Z}} \mapsto G$.
For any $\widetilde{\mbf{Z}}_{[0,t-1]} \in \mcl L^2(\Omega, \mcl F,\mu;\R^{\dimz})$, $\dimz\in\{ \dimu, \dimw,\dimy\}$, i.e., for any random variable trajectory of \eqref{eq:RVdynamics}, the corresponding PCE coefficient trajectories $ \trar{ \tilde{\pce z}}{t-1}^j, j\in \mbb{N}\cup \{\infty\} $ exist. Moreover, they correspond to $\pcecoe{g}{j}\in \R^{T-t+1}, j\in \mbb{N}\cup \{\infty\}$ such that \eqref{eq:mixed_funda} holds, cf. Corollary~\ref{cor:mixed_funda}. To conclude the proof  and similarly as before, we construct $G \in \splx{T-t+1} $ via the Galerkin projection of \eqref{eq:RVfunda} for all $j\in \mbb{N}\cup \{\infty\}$. 

\textit{Part 2):}
Note that Part 2) asserts that $g \mapsto \widetilde{\mbf{Z}}_{[0,t-1]}$.
With Part 2) of Lemma~\ref{lem:colEqui}, we have that for any $g \in \R^{T-t+1}$ there exists  $G \in  \splx{T-t+1} $ such that \eqref{eq:colEquiRV} holds. Furthermore, Part 1) of Lemma~\ref{lem:RVfundamental} gives that  $G$ determines a random variable trajectory of \eqref{eq:RVdynamics} by \eqref{eq:RVfunda}. 
\end{proof}

\subsection{Discussion}
\begin{figure*}
\begin{center}
	\includegraphics[width=0.63\textwidth]{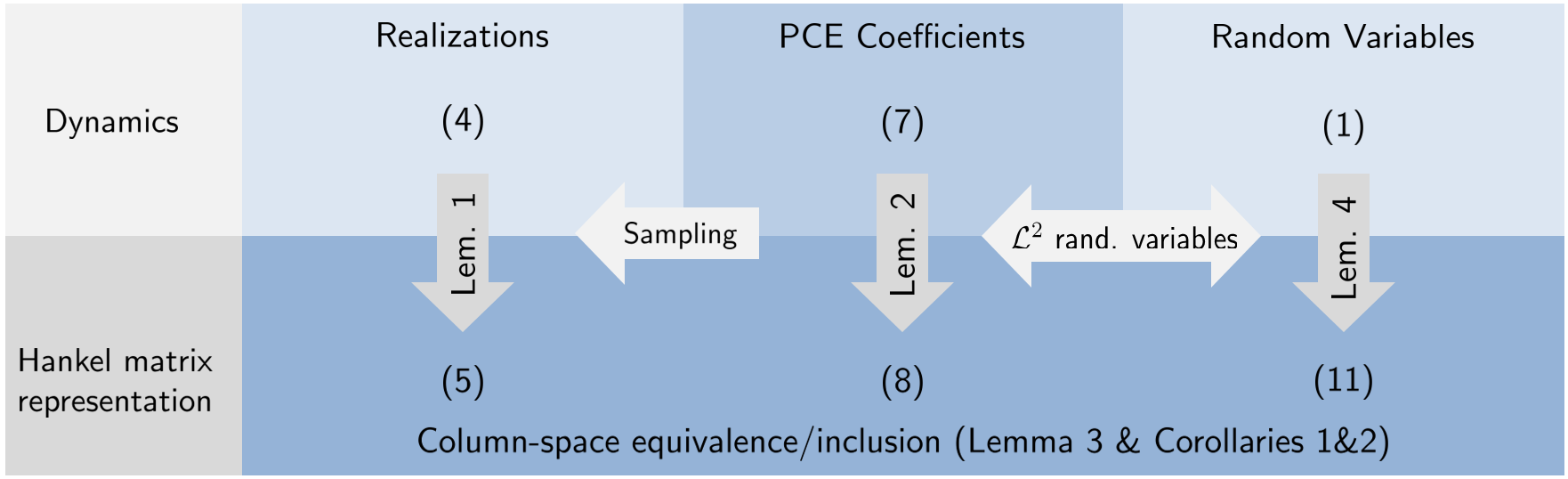}
\end{center}
\caption{Overview of the results of Section~\ref{sec:data_representation}.}\label{fig:ResSecIII}
\end{figure*}

Figure~\ref{fig:ResSecIII} summarizes the results derived above and their underlying assumptions. Indeed, Lemma~\ref{Lem:fundamental} and Lemma~\ref{lem:PCEfundamental} are immediate consequences of the original result by Willems et al.~\cite{Willems2005}. Observe that the assumed $\mcl L^2(\Omega, \mcl F,\mu;\R^{n_z})$ nature of the input and output trajectories of \eqref{eq:RVdynamics} allows to link the random variables and their PCE coefficient dynamics \eqref{eq:PCEcoesDynamics}. Moreover, note that by sampling $\phi^j(\xi(\omega))$ the PCE solutions---i.e., evaluating the basis functions $\phi^j$ for different argument realizations $\xi(\omega)$, $\xi \in \splx{n_\xi}$---one can go from PCE solutions to realizations. 

Besides these technicalities, crucial aspects are as follows:  
First,	as the estimation of PCE coefficients from data requires rather large data sets, the immediate usability of the fundamental lemma for the PCE coefficients (Lemma~\ref{lem:PCEfundamental}) appears to be limited. Yet, the equivalence and inclusion properties of column spaces obtained in Lemma~\ref{lem:colEqui}, Corollary~\ref{cor:inclusion}, and Corollary~\ref{cor:mixed_funda} provide the pivotal link, i.e., one may safely use a Hankel matrix of realizations to compute PCE coefficient  trajectories. 

The fundamental lemma is also frequently used in input-state settings~\cite{DePersis19}. We note that our preceding developments can easily transferred to the case of measured states by considering $y=x$.
Moreover, it is worth to be remarked that persistency of excitation of the $\mbf{(u,w)}$ realization trajectories is a much weaker requirement than persistency of excitation of the $(\pce {u,w})^j$ PCE coefficient trajectories. This is because in many situations one will assume that the distribution of the  disturbance is constant while the disturbance realization will usually be persistently exciting. 
Finally, the stochastic fundamental lemma (Lemma~\ref{lem:RVfundamental})---upon once more exploiting the column-space equivalence and inclusion---shows that one can extend the fundamental lemma towards the stochastic system~\eqref{eq:RVdynamics}.  

\begin{remark}[Why polynomial chaos expansions?]
At this point, it is fair to ask whether and to which extent the proposed stochastic fundamental lemma relies on polynomial chaos. Indeed, it stands to reason that despite manifold consideration of PCE in the systems and control context, stochastic systems are much more frequently approached via statistical moments, via probability densities, or via scenario approaches \cite{farina13probabilistic,Weissel2009,Schildbach2014}.
However, underlying our developments leading to the stochastic fundamental lemma is a pivotal observation: The stochastic LTI system~\eqref{eq:RVdynamics}, the realization dynamics~\eqref{eq:RelizationDynamics}, and the PCE coefficient dynamics~\eqref{eq:PCEcoesDynamics}  are subject to \textit{the same system matrices}. The reason for this structural similarity is the linear nature of PCEs. Actually, the usual description of Gaussian random variables in terms of mean value/expectation and standard deviation ($=$ square-root of the covariance) is a \textit{nonlinear parametrization} in terms of statistical moments. 
Put differently, the dynamics of the statistical moments beyond expectation are structurally different to~\eqref{eq:RVdynamics}, \eqref{eq:RelizationDynamics}, and~\eqref{eq:PCEcoesDynamics}.

Consequently, it is the structural similarity of~\eqref{eq:RVdynamics}, \eqref{eq:RelizationDynamics}, and~\eqref{eq:PCEcoesDynamics} which enables the formulation of the stochastic fundamental lemma. 
However, we also remark that a full-fledged behavioral characterization of stochastic systems~\eqref{eq:RVdynamics}, their PCE representations~\eqref{eq:PCEcoesDynamics}, and the relation between both is still an open problem. In a follow-up to the present analysis we provide first results in this direction~\cite{tudo:faulwasser22f}.  \End 
\end{remark}

Another crucial observation is as follows: The usual form of a fundamental lemma is that the columns of the Hankel matrix constructed from past system trajectories span the linear subspace of all possible trajectories of the LTI system. However, in the stochastic setting if the Hankel matrix is constructed from the random variables directly, this usual form of the lemma does not necessarily hold,  cf. Part 2) of Lemma~\ref{lem:RVfundamental}.
Specifically, note that upon assuming finite and exact PCEs, applying Galerkin projection for $L \in \mbb{N}+$ PCE basis functions to \eqref{eq:RVfunda}, and using the equivalence of $\mathrm{colsp}\big(\Hankel_t(\trar{\pce{z}}{T-1}^j) \big)$ and $\mathrm{colsp}\big(\Hankel_t(\trar{z}{T-1} )\big)$ given by \eqref{eq:colEqui}, we obtain 
\[I_L\otimes \Hankel_t(\trar{z}{T-1}) \pce g^{[0,L-1]} = \tilde{\pce z}^{[0,L-1]}\]
to compute the vector $\pce g^{[0,L-1]}$, where $\otimes$ denotes the Kronecker product, $I_L$ denotes an identity matrix of size $L$, and $\pce g^{[0,L-1]}$ stacks the $\pce g^j$ into one vector. In contrast, Galerkin projection of \eqref{eq:RVfunda2} combined with column-space equivalence gives 
\[\mbf{1}_L\otimes \Hankel_t(\trar{z}{T-1}) \pce g = \tilde{\pce z}^{[0,L-1]},\]
 where $\mbf{1}_L$ is the $L\times 1$ vector of all $1$.

 We conclude our discussion with a simple example  illustrating why Part~2) of Lemma~\ref{lem:RVfundamental} does not admit an \textit{iff} statement.

\begin{example}  
 Consider the scalar stochastic system
$X_{k+1} = X_k+ U_k$ 
with past data given by the PCEs
\begin{align*}
  X_0 =  0 \phi^0 + 0 \phi^1,\quad &   U_0= 0 \phi^0 + 1 \phi^1, \\
  X_1 =  0 \phi^0 + 1 \phi^1,\quad &   U_1= 1 \phi^0 + 0 \phi^1, \\
  X_2 =  1 \phi^0 + 1 \phi^1, \quad&   U_2= 1 \phi^0 + 1 \phi^1.
\end{align*}
Note that the PCE coefficients of $\trar{U}{2}$ satisfy the persistency of excitation required by   Part~2) of Lemma~\ref{lem:RVfundamental}. We aim to find $g$  in \eqref{eq:RVfunda2} to represent
$ \widetilde{X}_0= 0 \phi^0 + 1 \phi^1,\quad    \widetilde{U}_0= 0 \phi^0 +  1 \phi^1$.
We obtain \eqref{eq:RVfunda2} as
\begin{equation}\label{eq:RVequation}
\begin{bmatrix}
X_0& X_1& X_2 \\
U_0&U_1&U_2
\end{bmatrix}    
 g = \begin{bmatrix}
\widetilde{X}_0\\
\widetilde{U}_0
\end{bmatrix}.
\end{equation}
After applying Galerkin projection onto the basis functions and stacking the projected equations  we obtain
$    M g= c$ with
    \[
M= 
    \begin{bmatrix}
0 & 0&  1  \\
   0 &  1 & 1  \\\hline
   0  &  1 & 1   \\
   1 & 0&  1  \\
    \end{bmatrix} \quad c = \begin{bmatrix}
       0\\0\\\hline 1\\1
    \end{bmatrix}
\]
where the upper block corresponds to $\phi^0$ and the lower one to $\phi^1$.

 By the Rouché–Capelli theorem, $ M g= c$ admits a solution $g$ if and only if the block matrix $[M|c]$ has the same rank as $M$. Observe that in the example above $\mathrm{rank}(M)=3$  and $\mathrm{rank}\left(\left[M|c\right]\right) = 4$. Thus, we conclude that \eqref{eq:RVequation} does not admit solutions $g\in \mathbb{R}^4$. \End
\end{example}

%% file: Resource/Sec4_StochasticOCPs.tex
\section{Data-Driven Stochastic Optimal Control}\label{sec:stochastic_OCPs}
The previous section introduces the data-driven representations for stochastic LTI systems. Precisely, the stochastic LTI system~\eqref{eq:RVdynamics} and the PCE coefficient dynamics~\eqref{eq:PCEcoesDynamics} are both linked to the realization data from~\eqref{eq:RelizationDynamics}, cf. Corollary~\ref{cor:mixed_funda} and Lemma~\ref{lem:RVfundamental}. In this section, we turn toward using our results for stochastic optimal control, i.e. the data-driven reformulation of Problem~\ref{Problem1}. Moreover, we briefly discuss the estimation of past process  disturbance realizations and propose a conceptual framework for data-driven stochastic predictive control.

\subsection{Applying PCE to the Stochastic OCP}\label{sec:PCE_OCP}
Before reformulating Problem~\ref{Problem1} in a data-driven fashion, we recall its PCE-based reformulation. 
To this end, we assume that exact PCEs for the initial condition and for the process disturbances are known.
Put differently, for known distributions of the initial condition $\Xini$ and the  disturbance $W_k$, one should (if possible) choose the bases such that their PCE representations are exact, e.g., they follow Table~\ref{tab:askey_scheme}. 
Moreover, to obtain exact PCEs for the optimal solution $\mbf{(X^\star,U^\star,Y^\star)}_{[0,N-1]}$ in Problem~\ref{Problem1} the underlying basis is constructed accordingly.

\begin{assumption}[Exact PCEs for $X_\text{ini}$ and  $W_k$]\label{ass:exact_Ini}
The initial condition $\Xini$ and all i.i.d. $W_{k}$, $k\in \I_{[0,N-1]}$ in Problem \ref{Problem1} admit exact PCEs, cf. Definition ~\ref{def:exact_pce}, with $L_{\text{ini}}$ terms and $L_w$ terms, respectively, i.e. $\Xini = \sum_{j=0}^{L_{\text{ini}}-1} \xinipce^j \phi_\text{ini}^j$ and $W_k = \sum_{j=0}^{L_w-1} \pce{w}^j_k \phi_k^j$ for $k\in \I_{[0,N-1]}.$ \End
\end{assumption}
We note that for i.i.d. $W_{k}$, $k\in \I_{[0,N-1]}$, the bases $\phi_k$ are structurally identical but the realizations of the arguments used in the bases are independent. Hence, we distinguish them by the subscript $(\cdot)_k$.

To construct a basis in which $\mbf{(X^\star,U^\star,Y^\star)}_{[0,N-1]}$ admit exact PCEs, we introduce the projection $\Pi^\mathbb{L}$ of a random variable $Z \in \splx{\dimz}$ expressed in the basis $\{\phi^j\}_{j=0}^{\infty}$ onto a reduced basis containing a subset of functions $\mathbb{L} \subseteq \N \cup \{\infty\}$ as $\Pi^\mathbb{L}:~\mcl{L}^2 \to \mcl{L}^2$
\begin{equation} \label{eq:PIprojection}
\Pi^\mathbb{L}:~Z = \sum_{j=0}^{\infty}\pcecoe{z}{j} \phi^j \mapsto \widetilde{Z}= \sum_{j\in\mbb{L}}\pcecoe{z}{j} \phi^j.
\end{equation}
\begin{proposition}[Exact uncertainty propagation via PCE]\label{pro:no_truncation_error}
	Let Assumption~\ref{ass:exact_Ini} hold. Consider the optimal solution $\mbf{(X^\star,U^\star,Y^\star)}_{[0,N-1]}$ of Problem~\ref{Problem1} for some finite horizon $N \in \N^+$.
Then, the following statements hold:
 \begin{itemize}
 	\item[1)] Suppose that $ \trar{U^\star}{N-1} $ admits an exact PCE with respect to the finite-dimensional basis $ \{\phi^j\}_{j=0}^{L-1}$,
 	where
 	\begin{subequations}\label{eq:finite_basis}
 	\begin{align}\label{eq:terms}
	L &= L_{\text{ini}} + N(L_w-1) \in \N^+\vspace*{-0.2cm},  \\
 \label{eq:bases}
	\{\phi^j\}_{j=0}^{L-1} &= \left\{1, \{\phi_{\text{ini}}^j\}_{j=1}^{L_\text{ini}-1}, \bigcup_{k=0}^{N-1} \{\phi_{k}^j\}_{j=1}^{L_w-1} \right\},
\end{align}
 	\end{subequations}
 	with $1=\phi_{\text{ini}}^0(\omega)=\phi_{k}^0(\omega)$ for $k \in \I_{[0,N-1]}$ and all $\omega \in \Omega$. Then, $ \trar{Y^\star}{N-1} $ and $ \trar{X^\star}{N-1} $ also admit exact PCEs  with respect to the finite-dimensional basis~\eqref{eq:finite_basis}.
 	\item[2)] 
 	In absence of chance constraints \eqref{eq:chance_U}--\eqref{eq:chance_Y},  $\mbf{(X^\star,U^\star,Y^\star)}_{[0,N-1]} $ admit  exact PCEs  with respect to the finite-dimensional basis~\eqref{eq:finite_basis}.
 	\item[3)]  Consider chance constraints \eqref{eq:chance_U}--\eqref{eq:chance_Y} and  
 	an  orthogonal infinite-dimensional basis $ \{\phi^j\}_{j=0}^{\infty}$ whose first $L$ terms are given as \eqref{eq:finite_basis}. If $ \mbf{U^\star}_{[0,N-1]} $ admits non-zero PCE coefficients  $ \pce{u}^{\star, \tilde{j}}_{[0,N-1]} \neq 0 $ for some $\tilde{j} \geq L $, then  
 	\[ \mathbf{(\bar{X},\bar{U},\bar{Y})}_{[0,N-1]} = 
 	\Pi^{\mbb L}\left(\mbf{(X^\star,U^\star,Y^\star)}_{[0,N-1]}\right)\] is infeasible in Problem~\ref{Problem1} with $\mbb L = \I_{[0,L-1]}$. \End
 \end{itemize}
\end{proposition}
The proof follows ideas from \cite{muehlpfordt18comments}. It is given in Appendix~B. 

Observe that the basis $ \{\phi^j\}_{j=0}^{L-1}$~\eqref{eq:finite_basis} is the union of the independent bases $\phi_\text{{ini}}$ and $\phi_k$, $k \in \I_{[0,N-1]}$. 
 Hence, $\Xini$ and ${W}_k$, $k \in \I_{[0,N-1]}$ also admit exact PCEs in~\eqref{eq:finite_basis}, i.e.  $\Xini = \sum_{j=0}^{L-1} \tilde{\pce{x}}_{\text{ini}}^j \phi^j$ and $W_k = \sum_{j=0}^{L-1}\tilde{\pce{w}}_k^j \phi^j$ for all $k\in \I_{[0,N-1]} $. Precisely, considering the finite PCEs $\Xini = \sum_{j=0}^{L_{\text{ini}}-1} \xinipce^j \phi_\text{ini}^j$ and ${W}_k = \sum_{j=0}^{L_w-1} \pce{w}_k^j \phi_k^j$  in Assumption~\ref{ass:exact_Ini}, we have
\begin{equation}\label{eq:initial_PCE}
\tilde{\pce{x}}_{\text{ini}}^j =     \begin{cases}
\xinipce^j, & \forall j \in \I_{[0,L_\text{ini}-1]}\\
 0, & \forall j \notin \I_{[0,L_\text{ini}-1]}
\end{cases}, \tilde{\pce{w}}_k^j = \begin{cases}
\pce{w}_k^{\tilde{j} },&\forall j \in  \mbb{L} \\
0,  & \forall j \notin  \mbb{L}
\end{cases} 
\end{equation}
with $\mbb L = \{0\}\cup  \I_{[a,b]}$, $[a,b] = [L_\text{ini}+k(L_w-1),L_\text{ini}+(k+1)(L_w-1)-1]$, and $\tilde{j}= \max\{0,j-a+1\}$.

Moreover, due to the independence of $W_{k}$ at each time instant $k \leq N-1$, we note that, as the prediction horizon $N$ grows, the number of  required terms $L$ for exact PCEs in \eqref{eq:finite_basis} grows linearly.

\begin{remark}[Extension to non-i.i.d. settings]
If the disturbances $W_{k} \in \splx{\dimw}$ are independently but not identically distributed for different $k \in \N$, the considered bases of $W_{k}$ can be structurally different for $k \in \N$. That is, for all $k \in \N$, we consider $W_{k}$ admitting exact PCEs with $L_k$ terms in different bases $\phi_{k}$, i.e. $W_k = \sum_{j=0}^{L_k-1} \pce{w}_k^j \phi_k^j$. Then, the counterpart of basis \eqref{eq:finite_basis} reads $	\{\phi^j\}_{j=0}^{L-1} = \left\{1, \{\phi_{\text{ini}}^j\}_{j=1}^{L_\text{ini}-1}, \bigcup_{k=0}^{N-1} \{\phi_{k}^j\}_{j=1}^{L_k-1} \right\}$ with $L = L_{\text{ini}} + \sum_{k=0}^{N-1}(L_k-1)$. 
\vspace*{1mm}
	
If  $W_{k} \in \splx{\dimw}$ are dependently distributed for all $k \in \N$,  we may suppose that there exists a polynomial basis $\phi_w$ for all $W_{k}$, $k\in\N$ such that $W_{k}$, $k\in\N$ admit exact PCEs with at most $L_w$ terms, i.e. $W_k = \sum_{j=0}^{L_w-1} \pce{w}_k^j \phi_w^j$ for all $k\in\N$. Then we obtain $L =L_{\text{ini}} + L_w-1$ and the basis $\{\phi^j\}_{j=0}^{L-1} = \left\{1,\{\phi_{\text{ini}}\}_{j=1}^{L_\text{ini}-1},\{\phi_{w}\}_{j=1}^{L_w-1}\right\}$. 
\label{rem:dependentWk}
\End
\end{remark}

\begin{remark}[Filtered stochastic processes with PCE]
Consider the basis~\eqref{eq:finite_basis}, for $k \in \I_{[0,N-1]}$ we have 
\begin{equation*}\label{eq:pce_newbasis}
	U_{k} = \pce{u}_{k}^0 + \sum_{j=1}^{L_\text{ini}-1}\pce{u}_{k}^j \phi^j
	+ \sum_{i=0}^{N-1}\sum_{j=L_\text{ini}+i(L_w-1)}^{L_\text{ini}+(i+1)(L_w-1)-1}\pce{u}_{k}^j \phi^j .
\end{equation*}
Then, the causality/non-antipacitivity of the filtration \eqref{eq:filtration} implicitly imposes an additional constraint on the PCE coefficients of the input
\begin{equation} \label{eq:causality_input}
	\pce{u}_{k}^{j} = 0, j\in \I_{[L_\text{ini}+k(L_w-1),L-1]},\, k \in \I_{[0,N-1]}.
\end{equation}
We remark that the causality of $X_k$ and $Y_k$ trivially holds when \eqref{eq:causality_input} is imposed on the PCE coefficient dynamics~\eqref{eq:PCEcoesDynamics}. \vspace*{2mm}\End
\end{remark}

Similar to \cite{farina13probabilistic}, we consider a conservative approximation of chance constraints \eqref{eq:chance_U}--\eqref{eq:chance_Y} in Problem~\ref{Problem1} that reads 
\begin{equation}\label{eq:chance_reformulation}
 \mean[Z_k] \pm \sigma(\varepsilon_z)\sqrt{\var[Z_k]} \in \mbb{Z}, ~k\in \I_{[0,N-1]},~Z \in\{U,Y\},
\end{equation}
with
$\sigma(\varepsilon_z) = \sqrt{(2-\varepsilon_z)/\varepsilon_z}, z\in \{u,y\}$ \cite{farina13probabilistic}. Especially, for Gaussian random variables,  \eqref{eq:chance_U}--\eqref{eq:chance_Y} and \eqref{eq:chance_reformulation} are equivalent when $\sigma(\varepsilon_z)$ is chosen with respect to the standard normal table.  We remark that the square-root is applied component-wise in \eqref{eq:chance_reformulation}.  Now, we are ready to state the PCE reformulation of Problem \ref{Problem1}. 

\begin{problem}[Stochastic OCP in PCE coefficients]\label{Problem2} ~\\
	Let Assumption~\ref{ass:exact_Ini} hold. Consider the  finite-dimensional basis $\{\phi^j\}_{j=0}^{L-1}$  determined by \eqref{eq:finite_basis}, then the PCE reformulation of Problem \ref{Problem1} reads
\begin{subequations}\label{eq:PCE_SOCP}
\begin{gather}
     \min_{
\substack{
	\text{for }j \in \I_{[0,L-1]} \\
 \mbf{\pce{x}}^{j}_{[0,N-1]}\in \R^{N\dimx}\\
 \mbf{\pce{u}}^{j}_{[0,N-1]} \in \R^{N\dimu}\\
  \mbf{\pce{y}}^{j}_{[0,N-1]}\in \R^{N\dimy}
 }
 }
 \sum_{k=0}^{N-1} \sum_{j=0}^{L-1} \Big(\| \pcecoe{y}{j}_{k}\|^2_Q  +\|\pcecoe{u}{j}_{k}\|^2_R\Big) \langle \phi^j\rangle^2 \\
     \text{subject to }\quad  \forall j\in \set{I}_{[0,L-1]}, \, k \in \set{I}_{[0,N-1]} \nonumber \\
  \pcecoe{x}{j}_{k+1}=   A \pcecoe{x}{j}_{k}+B\pcecoe{u}{j}_{k} 
+E\tilde{\pce{w}}^{j}_{k},\quad   \pcecoe{x}{j}_0  =\tilde{\pce{x}}_{\text{ini}}^{j}, \label{eq:PCEdynamicinOCP}\\
 \pcecoe{y}{j}_{k}=   C \pcecoe{x}{j}_{k}+D\pcecoe{u}{j}_{k}, \label{eq:PCEdynamicinOCP_y}\\
       \pcecoe{u}{0}_{k}  \pm \sigma(\varepsilon_u)\sqrt{\sum_{j=1}^{L-1} {(\pcecoe{u}{j}_{k}})^2\langle \phi^j\rangle^2} \in \mathbb U, \label{eq:PCE_chance_U}\\
     \pcecoe{y}{0}_{k}   \pm \sigma(\varepsilon_y)\sqrt{\sum_{j=1}^{L-1} {(\pcecoe{y}{j}_{k})}^2\langle \phi^j\rangle^2} \in  \mathbb Y, \label{eq:PCE_chance_Y}\\
      \pce{u}_{k}^{j'}= 0,~ \forall j'\in \I_{[L_{\text{ini}}+k(L_w-1),L-1]} \label{eq:causality}.
    \end{gather} 
where $\tilde{\pce{x}}_{\text{ini}}^j$ and  $\tilde{\pce{w}}_k^j$, $j \in \I_{[0,L-1]}$, $k \in \I_{[0,N-1]}$ are given in \eqref{eq:initial_PCE}.
    \End
\end{subequations}
\end{problem}
The chance constraint reformulation from \eqref{eq:chance_U}--\eqref{eq:chance_Y}  to  \eqref{eq:PCE_chance_U}--\eqref{eq:PCE_chance_Y} follows by expressing the moment terms in \eqref{eq:chance_reformulation} via PCE, cf. \eqref{eq:MomentsPCE}. 

\subsection{Data-Driven Stochastic Optimal Control}

So far we have assumed knowledge of the initial state, i.e. the distribution of $\Xini$ and its finite and exact PCE $\sum_{j=0}^{L_{\text{ini}}-1} \xinipce^j \phi_\text{ini}^j$ (Assumption~\ref{ass:exact_Ini}). However, in the context of data-driven control, the information about the state is usually considered to be unknown. Thus, instead of assuming that $\Xini$ admits an exact PCE,  we move to an assumption on the past $\Tini$-steps input,  disturbance, and output trajectory of \eqref{eq:RVdynamics}, $(\mbf{U},\mbf{W},\mbf{Y})_{[-\Tini,-1]}$  with $\Tini$ being larger than the system lag~$l$. 

\begin{definition}[System lag \cite{markovsky06}] The lag of system \eqref{eq:RVdynamics} is defined as the smallest integer $l$ for which 
	\[\mcl{O}_{l}\doteq \left[ C^\top~(CA)^\top~\cdots~(CA^{l-1})^\top \right]^\top
	\]
	has full rank. \End
\end{definition}

\begin{assumption}[Exact PCEs of consistency data] \label{ass:exact_Ini_UY} ~\\
We assume that $\Tini$ is larger than the system lag $l$ and the past $\Tini$-steps input,  disturbance, and that the output  trajectories of \eqref{eq:RVdynamics}, $(\mbf{U},\mbf{W},\mbf{Y})_{[-\Tini,-1]}$,  admit exact PCEs  with $L_\text{ini}$ terms, i.e. $Z_k = \sum_{j=0}^{L_{\text{ini}}-1}\pce{z}_k^j\phi_{\text{ini}}^j$ for all $k \in \I_{[-\Tini,-1]}$ and $(Z,\pce{z}) \in \{(U,\pce{u}),(W,\pce{w}),(Y,\pce{y})\} $. Moreover, for all $k \in \I_{[0,N-1]}$, the i.i.d. $W_{k}$ admit  exact PCEs with $L_w$ terms, i.e. $W_k = \sum_{j=0}^{L_w-1} \pce{w}_k^j \phi_k^j$. \End
\end{assumption}
Likewise, we construct the finite-dimensional basis $ \{\phi^j\}_{j=0}^{L-1}$ as the union of $\phi_\text{{ini}}$ and $\phi_k$, $k \in \I_{[0,N-1]}$, cf.\eqref{eq:finite_basis}. Hence, for all $k \in \I_{[-\Tini,-1]}$, we have the exact PCEs $Z_k = \sum_{j=0}^{L-1}\tilde{\pce{z}}_k^j\phi^j$  in basis~\eqref{eq:finite_basis} with
\begin{equation}\label{eq:initial_PCE_UY}
	\tilde{\pce{z}}_k^j =     \begin{cases}
		\pce{z}_k^j, &\forall j \in \I_{[0,L_\text{ini}-1]}\\
		0,  & \forall j \notin \I_{[0,L_\text{ini}-1]}
	\end{cases}, \quad \pce{z} \in \{\pce{u},\pce{w},\pce{y}\},
\end{equation}
and the exact PCEs $W_k = \sum_{j=0}^{L-1}\tilde{\pce{w}}_k^j\phi^j$, $k \in \I_{[0,N-1]}$ as given in \eqref{eq:initial_PCE}.

Next, we give a  data-driven reformulation of \eqref{eq:PCE_SOCP}  exploiting  Corollary~\ref{cor:mixed_funda} and Lemma~\ref{lem:RVfundamental}.

\begin{problem}[Data-driven Stochastic OCP]\label{Problem3}~\\
Let Assumption~\ref{ass:exact_Ini_UY} hold, and let the finite-dimensional basis~$\{\phi^j\}_{j=0}^{L-1}$  be given  by \eqref{eq:finite_basis}. Suppose that  realization data $\trar{(u,w,y)}{T-1}$ of \eqref{eq:RVdynamics} is given with $\trar{(u,w)}{T-1}$ persistently exciting of order $\dimx +N+\Tini$. Then, the data-driven reformulation of Problem~\ref{Problem2} reads
\begin{subequations}\label{eq:H_PCE_SOCP}
\begin{align}
   & \min_{
\substack{ \text{for }j \in \I_{[0,L-1]} \\
 \pcecoe{u}{j}_{[-\Tini,N-1]}\in \R^{(N+\Tini)\dimu} \\
  \pcecoe{y}{j}_{[-\Tini,N-1]} \in \R^{(N+\Tini)\dimy}\\
  \pcecoe{g}{j}\in \R^{(T-N-\Tini+1)} \\
  }
 }
   \sum_{k=0}^{N-1} \sum_{j=0}^{L-1} \Big(\| \pcecoe{y}{j}_{k}\|^2_Q  +\|\pcecoe{u}{j}_{k}\|^2_R\Big)\langle \phi^j\rangle^2 \label{eq:H_PCE_SOCP_obj} \\
 &   \text{subject} \text{ to } \forall  j\in \set{I}_{[0,L-1]}, \nonumber   \\
 &\quad \begin{bmatrix}
 	\Hankel_{N+\Tini}(\trar{u}{T-1})\\
 	\Hankel_{N+\Tini}(\trar{y}{T-1})\\
 	\Hankel_{N+\Tini}(\trar{w}{T-1})\\
 \end{bmatrix}    
 \pcecoe{g}{j}= \begin{bmatrix}     
 	\pce{u}^j_{[-\Tini,N-1]}\\
 	\pce{y}^j_{[-\Tini,N-1]}\\     
 	\tilde{\pce{w}}^j_{[-\Tini,N-1]}\\
 \end{bmatrix} 
 , \label{eq:H_PCE_SOCP_hankel}
     \end{align} 
 \begin{align}
&
\pcecoe{u}{j}_{
	[-\Tini,-1]} =\tilde{\pce{u}}^{j}_{
	[-\Tini,-1]}, \quad \pce{y}^{j}_{
	[-\Tini,-1]} =   \tilde{\pce{y}}^{j}_{
	[-\Tini,-1]}, 
 \label{eq:intitialDeePC}
\\
&\hspace{3cm} \eqref{eq:PCE_chance_U} -\eqref{eq:causality}, \nonumber 
    \end{align} 
 where
 $\tilde{\pce{z}}_k^j $, $j \in \I_{[0,L-1]}$ are given in \eqref{eq:initial_PCE_UY} for $k\in \I_{[-\Tini,-1]}$, $\tilde{\pce{z}} \in \{\tilde{\pce{u}},\tilde{\pce{w}},\tilde{\pce{y}}\}$, and $\tilde{\pce{w}}_k^j $, $j \in \I_{[0,L-1]}$ are given in \eqref{eq:initial_PCE}  for $k\in \I_{[0,N-1]}$. \End
\end{subequations}
\end{problem}
\begin{remark}[Consistency condition via last $\Tini$ realizations] \label{rem:ini_via_realizations}
	We remark that one straight-forward specification of $(\mbf{U},\mbf{W},\mbf{Y})_{[-\Tini,-1]}$ in Assumption~\ref{ass:exact_Ini_UY} are the (observed/measured) realization values  $(\mbf{u},\mbf{w},\mbf{y})_{[-\Tini,-1]}$. In this case, we have $L_\text{ini} = 1$, i.e. $Z_k =z_k\phi_{\text{ini}}^0$ for all $k \in \I_{[-\Tini,-1]}$ and $(Z,\pce{z}) \in \{(U,\pce{u}),(W,\pce{w}),(Y,\pce{y})\} $ with $\phi_{\text{ini}}^0=1$.
	 \End
\end{remark}

\begin{remark}[Numerical solution with null-space projection]
For the sake of readability, in this remark we use  the short-hand notations $\Hankel_z$ and $\pce{z}^j$ instead of  $\Hankel_{N+\Tini}(\trar{z}{T-1})$ and $\pce{z}^j_{[-\Tini,N-1]}$ for $ (\mbf{z}, \pce{z})\in\{(\mbf{u}, \pce{u}),(\mbf{w}, \pce{w}),(\mbf{y}, \pce{y}) \}$.
If $\Hankel_w\in \R^{(N+\Tini) \dimw\times (T-N-\Tini+1)}$ is of full  row rank and  the PCE coefficients $\pce{w}^j$ are known, the null-space method can be employed to reduce the dimensionality of the decision variables, i.e. $\pcecoe{g}{j}$ in Problem \eqref{eq:H_PCE_SOCP}. To this end, we choose a matrix $M_w\in \R^{(T-N-\Tini+1)\times (T-(N+\Tini)(\dimw +1)+1)}$ whose columns span the null space of $\Hankel_w$. The core idea is to parametrize $\pcecoe{g}{j}$  in the equality constraint
$\Hankel_w \pcecoe{g}{j}=  \pce {w}^j$
as
\[	\pcecoe{g}{j} =   M_w {\pce h^j} + \Hankel_w^\rinv \pce {w}^j,\]
where ${\pce h^j}\in \R^{T-(N+\Tini)(\dimw +1)+1}$ and $\Hankel_w^\rinv$ denotes the Moore-Penrose inverse of $\Hankel_w$.
Thus, substitution of the above equation into  \eqref{eq:H_PCE_SOCP_hankel} yields a simplified and numerically favourable expression 
\[
		\begin{bmatrix} \Hankel_u \\ \Hankel_y\end{bmatrix} \left(M_w  {\pce h^j} + \Hankel_w^\rinv \pce {w}^j\right) = \begin{bmatrix} \pce{u}^j \\ \pce{y}^j \end{bmatrix}, \forall j\in \set{I}_{[0,L]}.\tag*{\End}
\]
\end{remark}\vspace*{0.1mm}

\begin{remark}[Solution with regularization]
	Observe that small data pertubations of the consistency conditions~\eqref{eq:intitialDeePC} might jeopardize the feasibility of OCP \eqref{eq:H_PCE_SOCP}. To overcome this issue, one adds the slack variable $s \in \R^{\dimy \Tini L\ini}$ to  \eqref{eq:intitialDeePC}, i.e.,
	\[
  \pce{y}^{j}_{
	[-\Tini,-1]} =   \tilde{\pce{y}}^{j}_{
	[-\Tini,-1]} + s^j.
	\]
Here with  slight abuse of notation $s^j$ denotes the corresponding elements of $s$ for $j \in \I_{[0,L_\text{ini}-1]}$.	Consequently, the penalty term  $\beta \|s\|_1$ with $\beta \gg 0$ is added to the objective.
 The use of slack variables is widely considered in deterministic data-driven predictive control. We refer to  \cite{Coulson2019} for insights into the one-norm penalization and to \cite{Berberich20} for the two-norm penalization.
\End
\end{remark}

\begin{remark}[Multiple-shooting implementation]~\\
We note that compared to the LTI model in \eqref{eq:PCEdynamicinOCP} the  equality constraint \eqref{eq:H_PCE_SOCP_hankel} increases the computational burden due to the large dense Hankel matrices. 
To overcome this issue, \cite{o2021data} suggests segmenting the prediction horizon into shorter intervals and using Hankel matrices of smaller dimension. Furthermore, the solution pieces in consecutive intervals  are coupled by continuity constraints. This idea resembles the classic concept of multiple shooting~\cite{Bock1984} in the data-driven setting. In~\cite{Ou23}, we tailor this concept to the data-driven stochastic OCP \eqref{eq:H_PCE_SOCP}. Moreover, combined with moment matching, one can  show that the dimension of the PCE basis and the number of decision variables can be reduced substantially. For details we refer to~\cite{Ou23}.\End
\end{remark}

\subsection{Equivalence of Stochastic OCPs}
The next result summarizes equivalence conditions for Problems~\ref{Problem1}--\ref{Problem3}. To this end, we define the set of the optimal input-output trajectories for each problem as follows:
\begin{align*}
	\mathcal S_1  &\doteq\left\{\mbf{(U^\star,Y^\star)}_{[0,N-1]}\,\middle |\,\begin{gathered} 
\exists \,\mbf{X}^\star_{[0,N-1]} \text{ s.t.}\\ 
\mbf{(X^\star,U^\star,Y^\star)}_{[0,N-1]} \\
 \text{is optimal  in  Problem~\ref{Problem1},}\\
	\text{ given  $\Xini$, $W_k$, $k \in \I_{[0,N-1]}$} 
\end{gathered}\right\}, 
\\
	\mathcal S_2  &\doteq\left\{\begin{gathered}
		 (\pce{u}^\star,\pce{y}^\star)^{j}_{[0,N-1]},\\ j \in \I_{[0,L-1]}
	\end{gathered}
		\,\middle |\,\begin{gathered} 
	\forall j \in \I_{[0,L-1]}, \,
\exists \, \pce{x}^{j,\star}_{[0,N-1]} \text{ s.t.}\\ 
	(\pce{x}^\star,\pce{u}^\star,\pce{y}^\star)^j_{[0,N-1]}\\
	\text{ is  optimal  in Problem~\ref{Problem2},}\\
	\text{ given   $\tilde{\pce{x}}^j_{\text{ini}}$, $\tilde{\pce{w}}^j_{k}$}, k \in \I_{[0,N-1]}
\end{gathered}\right\}, 
\\
		\mathcal S_3  &\doteq\left\{\begin{gathered} (\pce{u}^\star,\pce{y}^\star)^{j}_{[0,N-1]},\\ j \in \I_{[0,L-1]}
	\end{gathered}
	\,\middle |\,
	\begin{gathered} 
		\forall j \in \I_{[0,L-1]},  \\	((\pce{u}^{\star},\pce{y}^{\star})_{[-\Tini,N-1]},\pce{g}^{\star})^j,  \\
		\text{ is optimal  in Problem~\ref{Problem3},}\\
		\text{ given  $(\tilde{\pce{u}},\tilde{\pce{w}},\tilde{\pce{y}})^j_{[-\Tini,-1]}$,} \\  \tilde{\pce{w}}^j_{k}, k \in \I_{[0,N-1]} 
\end{gathered}\right\}.
\end{align*}
Note that $\mathcal S_1$ is a subset of $ \splx{(\dimu+\dimy)N}$ but $\mathcal S_2$ and $\mathcal S_3$ are subsets of $\R^{(\dimu+\dimy)NL}$. We consider the  map $\Psi:  \R^{(\dimu+\dimy)N L} \to   \splx{(\dimu+\dimy)N}$
\begin{align*}
	\Psi:
\pce{s}^{[0,L-1]} \mapsto \sum_{j=0}^{L-1}\pce{s}^j\phi^j \in  \splx{(\dimu+\dimy)N}.
\end{align*}

\begin{theorem}[Equivalence of stochastic OCPs]\label{thm:equivOCP} 
	Consider the stochastic LTI system \eqref{eq:RVdynamics}. Suppose that the pairs $(A,[B~E])$ and $(A,C)$ are, respectively, controllable and observable. Let Assumption~\ref{ass:exact_Ini_UY} hold and consider the finite-dimensional basis  $ \{\phi^j\}_{j=0}^{L-1}$  determined by  \eqref{eq:finite_basis}. Let the given realization data $\trar{(u,w,y)}{T-1}$ of \eqref{eq:RVdynamics} with $\trar{(u,w)}{T-1}$ be persistently exciting of order $\dimx +N+\Tini$. Then following statements hold:
\begin{itemize}
\item[(i)] There exists a $\Xini$ that admits an exact PCE with $L_\text{ini}$ terms such that $\mcl S_2 = \mcl S_3$.
\item[(ii)] Moreover, suppose that the chance constraints \eqref{eq:chance_U}--\eqref{eq:chance_Y} are box constraints,  i.e. $\mbb{Z} = [\underline z, \bar z]$ for $\mbb{Z} \in \{\mbb{U},\mbb{Y}\}$. If the reformulation from \eqref{eq:chance_U}--\eqref{eq:chance_Y}  to \eqref{eq:chance_reformulation} 
is exact, then for any $\Xini$ admitting an exact PCE with $L_\text{ini}$ terms, \[\mcl S_1 = \Psi(\mcl S_2)\]
holds, where $\Psi(\mcl S_2)$ is the element-wise image of $\mcl S_2$. \End
\end{itemize}
\end{theorem}
The proof is given in Appendix~C.

\subsection{Estimation of Past Process  Disturbance Realizations}
\label{sec:estimation}
So far we assumed the knowledge of past process  disturbance realizations $\trar{w}{T-1}$.
Thus, in this section, we focus on the  disturbance estimation for a simplified AutoRegressive with eXtra input (ARX) model as follows
 \begin{subequations}
 	\begin{equation}\label{eq:ARX}
 			Y_{k} = \widetilde A \mathbf{Y}_{[k-\Tini,k-1]} + \widetilde{B} \mathbf{U}_{[k-\Tini,k-1]} +W_{k-1}, 
 		\end{equation}
 where  $\widetilde A \in \R^{\dimy\times \Tini \dimy}$, $\widetilde B \in \R^{\dimy\times \Tini\dimu}$. Moreover, we assume that \eqref{eq:ARX} admits a minimal state-space realization as \eqref{eq:RVdynamics} with $D = 0$ and some $E\in \R^{\dimx \times \dimy}$.  For the construction of minimal state-space realizations of \eqref{eq:ARX} we refer to \cite[Lemma~2]{Sadamoto2022}.
  We remark that for the case of state measurements, i.e. $Y=X$ and $\Tini =1$, \eqref{eq:ARX} is equivalent to $X^+ =AX +BU +W$ which is a common model structure in the context of stochastic optimal control. Similar to Section~\ref{sec:problem_statement} we consider the process  disturbance $W_k$, $k\in \N$ to be i.i.d. random variables whose underlying distribution  is known.

 With given $(\mbf{u},\mbf{y})_{[-\Tini,-1]}$ and $w_k$, $k\in \I_{[-1,\infty)}$, we can specify the realization dynamics of \eqref{eq:ARX} as
\begin{equation}\label{eq:ARX_realization}
	y_{k} = \widetilde A \mathbf{y}_{[k-\Tini,k-1]} + \widetilde{B} \mathbf{u}_{[k-\Tini,k-1]}  + w_{k-1}, ~k \in \N.
\end{equation}
\end{subequations}

\begin{proposition}[Consistency of  realization data]\label{pro:noise_consistency}
	 Consider  realization data $(\mbf{u},\mbf{w},\mbf{y})_{[0,T]}$ of \eqref{eq:ARX_realization} for  given $(\mbf{u},\mbf{y})_{[-\Tini,-1]}$ and  unknown  disturbance realizations $w_{k}$, $k\in \I_{[-1,T]}$. Let  $I_T$ be an identity matrix of size $T$ and 
	 \begin{equation}\label{eq:matrixS}
	 S \doteq  \begin{bmatrix}
	 	\Hankel_{\Tini}\left(\mbf{y}_{[1-\Tini,T-1]}\right) \\
	 	\Hankel_{\Tini}\left(\mbf{u}_{[1-\Tini,T-1]}\right)
	 \end{bmatrix} \in \R^{\Tini(\dimu+\dimy) \times T},
	\end{equation} then the input,  disturbance, and output realizations satisfy
	\begin{equation}\label{eq:leftkern_ARX}
		\left( \Hankel_{1}\left(\mbf{y}_{[1,T]}\right)-\Hankel_{1}\left(\mbf{w}_{[0,T-1]}\right) \right)\left( I_T- S^\rinv S\right) =0,
	\end{equation}
 where $S^\rinv$ denotes the Moore-Penrose inverse of $S$. \End
\end{proposition}
\begin{proof}
	By horizontally stacking \eqref{eq:ARX_realization} for $k \in \I_{[1,T]}$, we have
	\begin{equation*}
		\begin{aligned}
			\Hankel_{1}\left(	\mbf{y}_{[1,T]}\right)& =\left[\widetilde A~\widetilde{B}\right] \begin{bmatrix}
				\Hankel_{\Tini}\left(\mbf{y}_{[1-\Tini,T-1]}\right) \\
				\Hankel_{\Tini}\left(\mbf{u}_{[1-\Tini,T-1]}\right)
			\end{bmatrix}
			  +\Hankel_{1}\left(\mbf{w}_{[0,T-1]}\right) \\
			&	=\left[\widetilde A~\widetilde{B}\right] 	S+\Hankel_{1}\left(\mbf{w}_{[0,T-1]}\right) .
		\end{aligned}
	\end{equation*}
	Using  $SS^\dagger S = S$ we obtain
	\[		\Hankel_{1}\left(	\mbf{y}_{[1,T]}\right)	=\left[\widetilde A~\widetilde{B}\right]  S S^\dagger S+\Hankel_{1}\left(\mbf{w}_{[0,T-1]}\right).\]
After substituting $\left[\widetilde A~\widetilde{B}\right]  S $ with  $\Hankel_{1}\left(	\mbf{y}_{[1,T]}\right)- \Hankel_{1}\left(\mbf{w}_{[0,T-1]}\right)$, we have
	\begin{align*}
		\Hankel_{1}&\left(	\mbf{y}_{[1,T]}\right)\\&= 	\left( \Hankel_{1}\left(\mbf{y}_{[1,T]}\right)-\Hankel_{1}\left(\mbf{w}_{[0,T-1]}\right) \right)S^\dagger S + \Hankel_{1}\left(\mbf{w}_{[0,T-1]}\right),
\end{align*}
	which is equivalent to \eqref{eq:leftkern_ARX}.
\end{proof}
\begin{corollary}\label{cor:estimate_system}
Consider  realization data $(\mbf{u},\mbf{w},\mbf{y})_{[0,T]}$ of \eqref{eq:ARX_realization} for  given $(\mbf{u},\mbf{y})_{[-\Tini,-1]}$ and  unknown  disturbance realizations $w_{k}$, $k\in \I_{[-1,T]}$. Then, for any $(\mbf{u},\hat{\mbf{w}},\mbf{y})_{[0,T]}$ satisfying \eqref{eq:leftkern_ARX}, there exist $\widehat{A} \in \R^{\dimy\times \Tini \dimy}$ and $\widehat{B} \in \R^{\dimy\times \Tini\dimu}$ such that $(\mbf{u},\hat{\mbf{w}},\mbf{y})_{[0,T]}$ satisfies the system equation
\begin{equation*}
	y_{k} =\widehat{A}  \mathbf{y}_{[k-\Tini,k-1]} + \widehat{B} \mathbf{u}_{[k-\Tini,k-1]}  + \hat{w}_{k-1}. \tag*{\End}
\end{equation*}
\end{corollary}
\begin{proof}
We note that $\left( I_T- S^\rinv S\right)$ is the orthogonal projector onto the kernel of $S$. Thus, $(\mbf{u},\hat{\mbf{w}},\mbf{y})_{[0,T]}$ satisfying \eqref{eq:leftkern_ARX} implies that its projection onto the kernel of $S$ is zero; in other words, it lies in the image space of $S$. Thus, for each tuple $(\mbf{u},\hat{\mbf{w}},\mbf{y})_{[0,T]}$ satisfying \eqref{eq:leftkern_ARX}, there exists a matrix $M\in \R^{\dimy\times \Tini (\dimy+\dimu)}$ such that 
\[
\Hankel_{1}\left(\mbf{y}_{[1,T]}\right)-\Hankel_{1}\left(\hat{\mbf{w}}_{[0,T-1]}\right) = M S.
\]
Then, for each row of the above equation, we have
\[
	y_{k} = M \begin{bmatrix}
	\mathbf{y}_{[k-\Tini,k-1]} \\
    \mathbf{u}_{[k-\Tini,k-1]}
	\end{bmatrix}  + \hat{w}_{k-1}, \quad \forall k \in \I_{[1,T]}. 
\]
By splitting $M$ into $[\widehat{A}~\widehat{B}]$ the assertion follows.
 \end{proof}
As shown in Corollary~\ref{cor:estimate_system}, any $(\mbf{u},\hat{\mbf{w}},\mbf{y})_{[0,T]}$ satisfying \eqref{eq:leftkern_ARX} implicitly determines an LTI system. Thus,  the usual Hankel matrix equations, i.e. \eqref{eq:mixed_funda} and \eqref{eq:RVfunda},  stated in Lemma~\ref{lem:RVfundamental} and Corollary~\ref{cor:mixed_funda}, hold for $(\mbf{u},\hat{\mbf{w}},\mbf{y})_{[0,T]}$. 

Notice that \eqref{eq:leftkern_ARX} admits infinitely many solutions  $\mbf{w}_{[0,T-1]}$. However, we can utilize the knowledge about the distribution to formulate the \textit{ maximum likelihood estimate}
\begin{subequations}\label{eq:noise_estimation}
\begin{align}\label{eq:noise_estimation_MLE}
	\hat{\mbf{w}}_{[0,T-1]}= \argmin_{\mbf{w}_{[0,T-1]}} -\sum_{k=0}^{T-1} \log p_w(w \inst{k}), ~ \text{s.t. } \eqref{eq:leftkern_ARX},
\end{align}
where $w\inst{k}$ is the realization of $W$ at time  $k$, and $p_w$ is the probability density function of the i.i.d. $W_k$, $k \in \N$. 

Alternatively, one can rely on the least-squares estimate 
\begin{align}\label{eq:noise_estimation_LSE}
	\hat{\mbf{w}}_{[0,T-1]} = \argmin_{\mbf{w}_{[0,T-1]} } \|  \mbf{w}_{[0,T-1]} -\mean[\mbf{W}_{[0,T-1]}]\|^2, ~ \text{s.t. } \eqref{eq:leftkern_ARX},
\end{align}
which admits the closed-form solution
\begin{equation*}
	\Hankel_1(\hat{\mbf{w}}_{[0,T-1]}) = \left(\Hankel_{1}\left(\mbf{y}_{[1,T]}\right) - \Hankel_1(\mean[\mbf{W}_{[0,T-1]}]) \right) ( I_T-S^\rinv S).
\end{equation*} 
\end{subequations}

Specifically, for Gaussian-distributed disturbances, i.e., $w\inst k \sim \mathcal{N}(0,\sigma^2)$, $k \in \set{I}_{[0,T-1]}$, the maximum likelihood estimate~\eqref{eq:noise_estimation_MLE} is equivalent to the least-squares estimate~\eqref{eq:noise_estimation_LSE}. For uniformly distributed $w\inst{k} \sim \mathcal{U}([-a,a]), k \in \set{I}_{[0,T-1]}$, the maximum likelihood estimate \eqref{eq:noise_estimation_MLE} does not admit a unique solution due to the constant probability  density function. Therefore, we apply \eqref{eq:noise_estimation_LSE} in this case.

\subsection{Conceptual Framework for Data-Driven Stochastic Predictive Control}

Combining the results of the previous sections, we propose the following data-driven stochastic MPC scheme based on OCP~\eqref{eq:H_PCE_SOCP} as summarized in Algorithm~\ref{alg:datadrivenSMPC}. In the offline data collection and pre-processing phase, random inputs $\trar{u}{T}$ are generated to obtain
$\trar{y}{T}$. The  disturbance realizations $\trar{w}{T-1}$ are assumed to be given, or they can be estimated by \eqref{eq:noise_estimation} when the simplified model~\eqref{eq:ARX} is considered. During the online optimization phase, we consider recursively solving the data-driven stochastic OCP \eqref{eq:H_PCE_SOCP} with given last $\Tini$ realizations $\mbf{(u,w,y)}_{[k-\Tini,k-1]}$, cf. Remark~\ref{rem:ini_via_realizations},  at the current time instant $k$. 
 	
 Observe that to specify the consistency data $\mbf{(u,w,y)}_{[k-\Tini,k-1]}$, one needs to measure or estimate the realization of past disturbances online. For the case of unknown disturbance realizations, we do the online estimation of $w_{k-1}$ by appending $(u,y)_{[k-\Tini,k-1]}$ to the offline input-output data $S$ in \eqref{eq:matrixS}. That is, we employ the disturbance estimation~\eqref{eq:noise_estimation} with respect to $S'=[S|(u,y)_{[k-\Tini,k-1]}]$ and the last element of the estimated disturbance sequence gives $w_{k-1}$.
 Then, we solve the data-driven stochastic OCP \eqref{eq:H_PCE_SOCP} for  $(\pcecoe{u}{j,\star}_{[-\Tini,N-1]}, \pcecoe{y}{j,\star}_{[-\Tini,N-1]},\pcecoe{g}{j,\star})$, $j\in \I_{[0,L-1]}$. Observe that the PCE coefficient  $\pcecoe{u}{0,\star}_0$ on $\phi^0=1$
is applied to system  \eqref{eq:RVdynamics} as the current feedback.

A detailed analysis of the closed-loop properties of the proposed data-driven stochastic predictive control framework is beyond the scope of the present paper and postponed to future work. Instead we demonstrate its efficacy via examples.

\begin{algorithm}[t]
        \caption{ Data-driven  stochastic predictive control} \label{alg:datadrivenSMPC}
\algorithmicrequire $T, N, L \in \N^+ $, $\mbf{(u,y)}_{[-\Tini,-1]}$, $k \leftarrow 0$\\
\textbf{Data collection and pre-processing (offline):}
\begin{algorithmic}[1]
\State Select uniformly random distributed $\trar{u}{T}\in \mbb{U}^{T+1}$\label{step:pre_1}
\State Apply $\trar{u}{T}$ to system \eqref{eq:RVdynamics}, record  $\trar{y}{T}$ 
\State Measure  / estimate $\trar{w}{T-1}$ by \eqref{eq:noise_estimation}
\State If $\trar{(u,w)}{T-1}$ is persistently of exciting of order less than $N+n_x$, go to Step \ref{step:pre_1}, else go to  Step  \ref{step:pre_5}
\State Construct \eqref{eq:H_PCE_SOCP} with $\trar{(u,w,y)}{T-1}$ 
\label{step:pre_5}
\end{algorithmic}
\textbf{Predictive control loop (online):}
\begin{algorithmic}[1]
\State   Measure $\mbf{(u,y)}_{[k-\Tini,k-1]}$, estimate $w_{k-1}$ by \eqref{eq:noise_estimation} \label{step:MPC_0}
\State  Solve $\eqref{eq:H_PCE_SOCP}$ with respect to  $\mbf{(u,w,y)}_{[k-\Tini,k-1]}$ \label{step:MPC_1}
\State  Apply $u_{k} = \pcecoe{u}{0,\star}_0$ to system \eqref{eq:RVdynamics}\label{step:MPC_2}
\State  $k\leftarrow k+1$, go to Step \ref{step:MPC_0}
\end{algorithmic}
\end{algorithm}

%% file: Resource/Sec5_Numericalexample.tex
\section{Numerical Examples}\label{sec:simulation}
We consider two examples for data-driven stochastic OCPs and predictive control via PCE. Moreover, the least-square estimate~\eqref{eq:noise_estimation_LSE} is employed to reconstruct the disturbance realizations. In the first (scalar) example, we consider Gaussian distributed disturbance as well as uniformly distributed disturbance. The second example considers discrete-time stochastic predictive control for an aircraft. Furthermore, we also verify that the proposed disturbance estimation performs well in open-loop OCPs and predictive control problems. To implement the numerical examples in \texttt{julia}, we rely on the toolboxes \texttt{PolyChaos.jl}~\cite{muehlpfordt19polychaos} and \texttt{JuMP.jl}~\cite{Dunning17}.

\subsection{Scalar Dynamics}
We consider the scalar OCP from \cite{grune13economic} and  its stochastic extension proposed in \cite{Ou21}. The dynamics are
	$X_{k+1} = 2X_k + U_k +W_k$,
where $W_k$ is a random disturbance with known distribution, and $X\ini$ follows the uniform distribution $\mcl U([0.2, 1.0])$. The matrices $Q$ and $R$ in \eqref{eq:stochasticOCP_obj} are $Q=0$ and $R=1$ while the state chance constraint reads
$\mathbb{P}[X \in \set{X}] \geq 1- \varepsilon_x$ with $\set{X} = [-2,2]$ and $\varepsilon_x = 0.2$. Thus, we have $\sigma(\varepsilon_x)=3$ in  \eqref{eq:chance_reformulation}. We solve this example as one open-loop OCP  with horizon $N=25$.

\subsubsection*{Gaussian distributed disturbance}

We suppose that no recorded disturbance data $\mathbf{\lowercase{w}}$ is available and that for all $k\in \mbb{I}_{[0, N-1]}$ the disturbances $W_k$, $k\in \N$ are i.i.d.  Gaussian $\mcl N(0,0.5^2)$. We apply \eqref{eq:noise_estimation_LSE} with  $1000$ samples to reconstruct the disturbance realizations. To this end, we use the feedback law $u_k=-1.1x_k + v_k$, where $v_k$ is randomly uniformly sampled from $[-0.1,0.1]$ for $k \in \I_{[0,999]}$, during the data collection phase. Moreover, we construct the Hankel matrices from the first 100 recorded input/state and estimated noise. This way, we ensure that the data is persistently exciting and the magnitude of state will not increase exponentially. We remark that, in the general context of data-driven control, the design of persistently exciting input sequences for unstable systems is still an open problem.

We choose Legendre polynomials as PCE basis with $L\ini=2$ for $X\ini$ and Hermite polynomials with $L_w=2$ for $W_{k}$ such that Definition \ref{def:exact_pce} is satisfied as shown in Table \ref{tab:askey_scheme}.  Consequently, we obtain $L=27$ from  \eqref{eq:terms} and the exactness of the PCEs is ensured. The solutions of the open-loop OCP~\eqref{eq:H_PCE_SOCP}  are depicted in Figure~\ref{fig:scalar_compare_normal}. Therein we compare the solution with estimated disturbance realizations to the one with exact knowledge of disturbance realizations. For the sake of simplicity, only the expected value and variance of $X$ and $U$ are plotted instead of PCE coefficients. It can be seen that the solution with estimated disturbance realizations from \eqref{eq:noise_estimation_LSE} matches the one with exact knowledge of disturbance realizations well.
More precisely, the maximum difference of the solutions of OCP~\eqref{eq:H_PCE_SOCP} with disturbance measurement and estimation in terms of the first two moments of $X$ and $U$ is $3.71\cdot 10^{-3}$. That is, \eqref{eq:H_PCE_SOCP} with estimated disturbance \eqref{eq:noise_estimation_LSE} provides a slightly suboptimal solution without much performance loss.
Furthermore, the maximum difference of the solutions of OCP~\eqref{eq:PCE_SOCP} and \eqref{eq:H_PCE_SOCP} with disturbance measurement in terms of the first two moments of $X$ and $U$ is $4.83\cdot 10^{-5}$, which shows that \eqref{eq:H_PCE_SOCP} is equivalent to model-based stochastic control if exact knowledge of disturbance realizations is available. 

\begin{figure}[t!]
	\begin{center}
		\includegraphics[width=8.4cm]{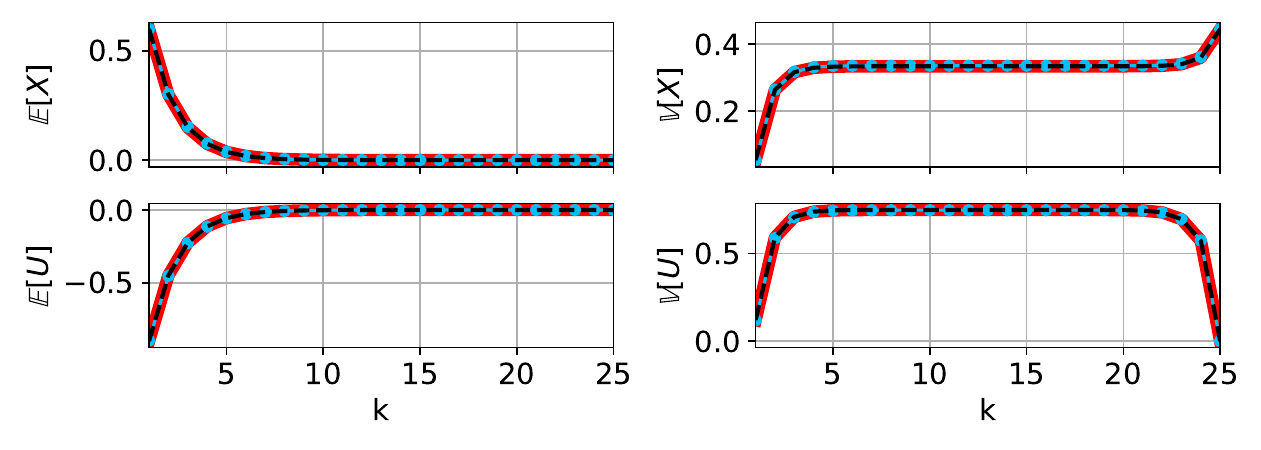}
		\caption{Scalar example with Gaussian disturbance. Red-solid line: solution of \eqref{eq:H_PCE_SOCP} with disturbance measurement; blue-solid line with circle marker: solution of \eqref{eq:H_PCE_SOCP} with disturbance estimation; black-dashed line: solution of \eqref{eq:PCE_SOCP}.
		} 
		\label{fig:scalar_compare_normal}
	\end{center}
\end{figure}

\begin{figure}[t!]
	\begin{center}
		\includegraphics[width=8.4cm]{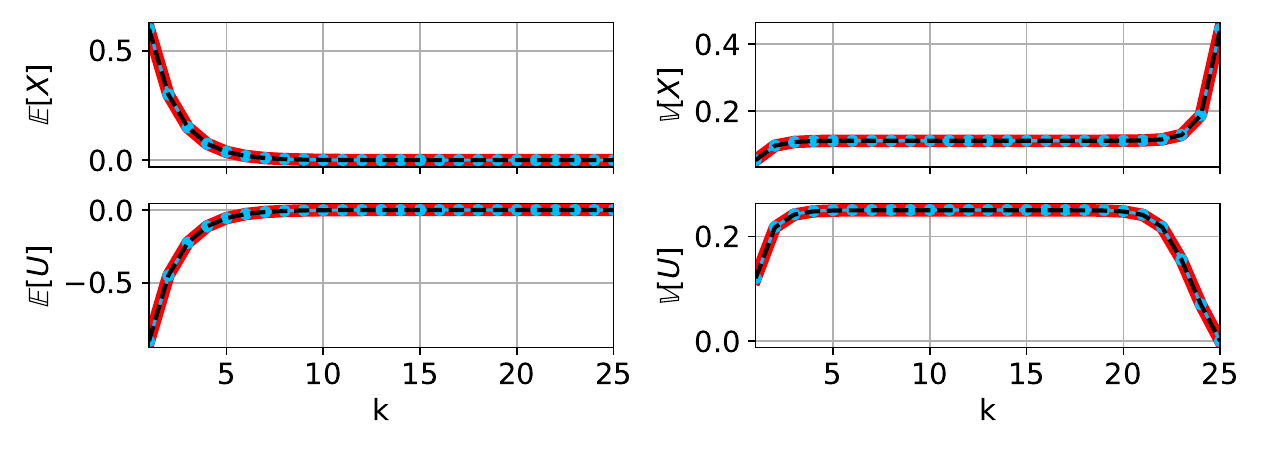}
		\caption{Scalar example with uniform disturbance. Red-solid line: solution of \eqref{eq:H_PCE_SOCP} with disturbance measurement; blue-solid line with circle marker: solution of \eqref{eq:H_PCE_SOCP} with disturbance estimation; black-dashed line: solution of  \eqref{eq:PCE_SOCP}.
		} 
		\label{fig:scalar_compare_uniform}
	\end{center}
\end{figure}

To further illustrate the minuscule differences of the solutions in Figure~\ref{fig:scalar_compare_normal},  we compare the underlying optimal value functions and the PCE solutions $\pce{x}_k^{j}$, with $k \in \I_{[0,N-1]}$ and $j \in \I_{[0,L-1]}$ in Table~\ref{tab:scalar_open_compare}. Here $V_N^\text{m}$, $V_N$, and $V_N^\text{e}$ are the optimal value functions of the model-based OCP~\eqref{eq:PCE_SOCP} (superscript $\cdot^\text{m}$) and, respectively, of the data-driven OCP~\eqref{eq:H_PCE_SOCP} with disturbance measurement or estimation  (superscript $\cdot^\text{e}$). Moreover, we introduce
	\[
	\delta_x\doteq \max_{j \in \I_{[0,L-1]}, k \in \I_{[0,N-1]}} \left|\pce{x}_k^{j}-\pce{x}_k^{j,{\text{m}}}\right|
	\]
	and similarly  $\delta_x^{\text{e}}$ to denote the maximum difference of OCP~\eqref{eq:PCE_SOCP} and OCP~\eqref{eq:H_PCE_SOCP} with disturbance measurement or estimation respectively.   As shown in the first row of Table~\ref{tab:scalar_open_compare}, the solutions of OCP~\eqref{eq:PCE_SOCP} and OCP~\eqref{eq:H_PCE_SOCP} with disturbance measurement are almost the same, while OCP~\eqref{eq:H_PCE_SOCP} with disturbance estimation provides also a slightly less accurate result. 
 
The observations above  are inline with  Theorem~\ref{thm:equivOCP}, i.e., OCP~\eqref{eq:H_PCE_SOCP} is the data-driven equivalent of OCP~\eqref{eq:PCE_SOCP} when persistency of excitation is satisfied and when the disturbance distributions as well as previous disturbance realizations  are known.

\subsubsection*{Uniformly distributed disturbance}
To verify the performance of the disturbance estimation~\eqref{eq:noise_estimation_LSE} with respect to uniformly distributed disturbance, we suppose $W_k$ follows a uniform distribution $\mcl U([-0.5,0.5])$. Note that the uniformly distributed disturbance admits an exact PCE with $L_w=2$ terms in the basis of the Legendre polynomials
and thus $L= 27$ as before. We also record state/input trajectories with horizon $T=1000$ and reconstruct the disturbance realizations by \eqref{eq:noise_estimation_LSE}. The solutions of the OCPs are illustrated in Figure \ref{fig:scalar_compare_uniform}. As one can see, the  estimation \eqref{eq:noise_estimation_LSE} also handles uniformly distributed disturbance. Similarly, the maximum difference of the all solutions in terms of the first two moments of $X$ and $U$ is $1.05\cdot 10^{-3}$.  We also observe the similarity of the solutions for uniformly distributed  $W_K$, cf. the second row of Table~\ref{fig:scalar_compare_normal}.

\begin{table}[t!]
	\caption{Comparison of the solutions of \eqref{eq:PCE_SOCP}, \eqref{eq:H_PCE_SOCP} with disturbance measurement, and \eqref{eq:H_PCE_SOCP}  with disturbance estimation of the scalar example for both Gaussian and uniform disturbances. }
	\label{tab:sclarcost}
		\centering
		\begin{adjustbox}{width=\columnwidth,center}
		\begin{tabular}{cccccc}
			\toprule
			 Cases &  $V_N^\text{m}$ $[-]$ &  $V_N$ $[-]$ & $V_N^\text{e}$ $[-]$ &$ \delta_x$  $[-]$ &$\delta_x^{\text e}$  $[-]$\\
		\midrule
	Gaussian & 17.993 & 17.993&17.961 & $3.106\times 10^{-5}$ & $3.269\times 10^{-3}$\\
			Uniform & \phantom{1}6.588 & \phantom{1}6.588 & \phantom{1}6.605 & $9.121\times 10^{-5}$ & $1.419\times 10^{-3}$ \\
			\bottomrule
		\end{tabular} \label{tab:scalar_open_compare}
		\end{adjustbox}
\end{table}

\subsection{Aircraft Example}

As a second example, we use the LTI aircraft model 
given in \cite{maciejowski02predictive} exactly discretized with sampling time $t_s=0.5~\text{s}$.
The system matrices are
\begin{subequations}
\begin{align*}
	A &= \begin{bmatrix} \phantom{-0}0.240 & \phantom{0}0\phantom{0.} & 0.179 & 0 \\
		-\phantom{0}0.372 & \phantom{0}1\phantom{.0} & 0.270 & 0 \\
		-\phantom{0}0.990 & \phantom{0}0\phantom{.0} & 0.139 & 0 \\
		-48.9\phantom{00} & 64.1 & 2.40\phantom{0} & 1\\
	\end{bmatrix},~
	B = \begin{bmatrix} -1.23 \\ -1.44 \\ -4.48 \\ -1.8\phantom{0} \end{bmatrix},
	\end{align*}
	\begin{align*}
	C &= \begin{bmatrix} \phantom{-00}0\phantom{.0} & \phantom{00}1\phantom{.0} & 0 & 0 \\
		\phantom{-00}0\phantom{.0} & \phantom{00}0\phantom{.0} & 0 & 1 \\
		-128.2 & 128.2 & 0 & 0\end{bmatrix},
	\hspace{21pt}D = 0_{3\times 1}.
\end{align*}
\end{subequations}
We consider a Gaussian disturbance $W_k$ affecting the input-output dynamics in form of \eqref{eq:ARX}, where $W_k, k\in \N$ are i.i.d. vector-valued random variables with 
$\Sigma [W_k,W_k] = \text{diag}([ 10^{-4}, 16, 0.16])$.
The weighting matrices in the objective function are $Q=\text{diag}([3.2407, 1.3695, 7.9270])$ and $R = 5188.25$. A chance constraint is imposed on $Y^1$
\[
	\mathbb{P}[Y^1 \in \set{Y}^1] \geq 1- \varepsilon_y,
\]
where $\set Y^1 = [-0.349,0.349]$ and $ \varepsilon_y = 0.1$. Correspondingly, we find $\sigma(\varepsilon_y)=1.645$ according to the standard normal table. 
We compare four different data-driven schemes:
\begin{enumerate}[label=\Roman*]
		\item Algorithm~\ref{alg:datadrivenSMPC} with disturbance measurement \label{alg:meas}
		\item Algorithm~\ref{alg:datadrivenSMPC} with disturbance estimation	\label{alg:est}
		\item Algorithm~\ref{alg:datadrivenSMPC} with dist. estimation and truncated PCEs \label{alg:trun}
		\item data-driven deterministic predictive control with slack variables~~\cite{Coulson2019}. \label{alg:deter}
\end{enumerate}

We apply Algorithm~\ref{alg:datadrivenSMPC} with prediction horizon $N=10$ and exact output feedback, i.e., the realization of $Y_k$ is known upon solving each OCP. Similar to before, in the data collection phase we record input-output trajectories of $1000$ steps to estimate the disturbance realizations via~\eqref{eq:noise_estimation_LSE}. We use the first $90$ recorded inputs-outputs and estimated disturbances to  construct the Hankel matrices. 

To obtain an exact PCE for each component of $W_k$, we employ the Hermite polynomials component-wise such that $L_w=4$.  Since there is no uncertainty caused by initial condition, i.e. $L\ini=1$, we obtain from \eqref{eq:terms} the dimension of the overall PCE basis as $L=31$. Considering an initial condition close to $y\ini = [0,-400,0]^\top$ and using identical disturbance realizations, we compute the closed-loop responses for Schemes~\ref{alg:meas}-\ref{alg:est}. 

To speed up the computation, one may limit the number of terms in the PCEs of $(U,W,Y)$. Hence in Scheme~\ref{alg:trun}, we truncate the PCE of $W_k$ after the first $4$ terms. That is, at each time step $k$, we solve the OCP over the horizon $j\in [k, k+N-1]$ considering only  the first disturbance $W_k$ and setting the PCE coefficients $\mathsf w_j=0$ for $j>k$. 

For the sake of comparison, we also use Scheme~\ref{alg:deter} proposed in \cite{Coulson2019}, see Figure~\ref{fig:AircraftComparison}. Note that a low-rank approximation of the Hankel matrices via singular value decomposition and truncation is applied to Scheme~\ref{alg:deter} to filter the disturbance~\cite{Coulson2019}.\footnote{The low-rank approximation method is only mentioned in the extended arXiv version of \cite{Coulson2019}.} We observe that the system input and response trajectories for Scheme~\ref{alg:meas} and Scheme~\ref{alg:est} are almost identical. Moreover, Figure~\ref{fig:AircraftComparison} illustrates that Algorithm~\ref{alg:datadrivenSMPC} drives the system to origin with a better performance, i.e., with faster speed and better robustness to the disturbance than Scheme~\ref{alg:deter}. 

The computation times and the closed-loop costs $J^{\text{cl}}$ of the realized closed-loop trajectories in Figure~4 
are summarized in Table~\ref{tab:AircraftComparisonTime}. All considered schemes are compared by the mean and the Standard Deviation (SD) of the computation time of each OCP evaluated in the closed loop. Compared to the time needed for solving data-driven OCPs, the computational effort for online disturbance estimation for the consistency condition is negligible. As Table~\ref{tab:AircraftComparisonTime} shows, for Scheme~\ref{alg:trun}  the computation  is significantly accelerated while the suboptimality is minor. 
More precisely, evaluating the closed-loop realization trajectories over time, we compute the closed-loop costs defined as $J^{\text{cl}}\doteq\sum_{k=0}^{39}\left[\|y_k\|^2_Q+\|u_k\|^2_R\right]$ for different schemes in Table~\ref{tab:AircraftComparisonTime}. 
The performance loss due to the truncated PCEs, i.e. $(J^{\text{cl}}_{\text{III}}-J^{\text{cl}}_{\text{II}})/ J^{\text{cl}}_{\text{II}}$, is 4.43\% and the maximal input difference is $1.06\cdot 10^{-3}\,\SI{}{rad}$ for this specific example. An intuitive explanation for this phenomenon is that in the closed predictive control loop only the first element of the input solution is applied to the system. Thus, the next upcoming disturbance is dominant. Moreover, comparing the closed-loop costs of Schemes~\ref{alg:meas} and~\ref{alg:est}, we observe that in our simulations  Scheme~\ref{alg:est} with estimated disturbances performs slightly better than Scheme~\ref{alg:meas}. However, a detailed analysis of the performance and robustness of Algorithm~\ref{alg:datadrivenSMPC} with truncated PCEs and disturbance estimation is left for future research. 

Additionally, we sample $50$ sequences of disturbance realizations. The corresponding closed-loop realization trajectories of Scheme~\ref{alg:est} are shown in Figure~\ref{fig:AircraftTraj}. It can be seen that the chance constraint for  $Y^1$ is satisfied with a high probability.
Moreover, we sample a total of $1000$ sequences of disturbance realizations and initial conditions around $[0,-400,0]^\top$. Then we compute the corresponding closed-loop responses of Scheme~\ref{alg:est}. The time evolution of the (normalized) histograms of the output realizations $y^2$ at $k=0, 10, 20, 30, 40$ is shown in Figure~\ref{fig:AircraftDistEvolution}, where the vertical axis refers to the probability density of $Y^2$. As one can see, the proposed control scheme achieves a narrow distribution of~$Y^2$ around $0$.

\begin{figure}[t!]
	\begin{center}
		\includegraphics[width=8.4cm]{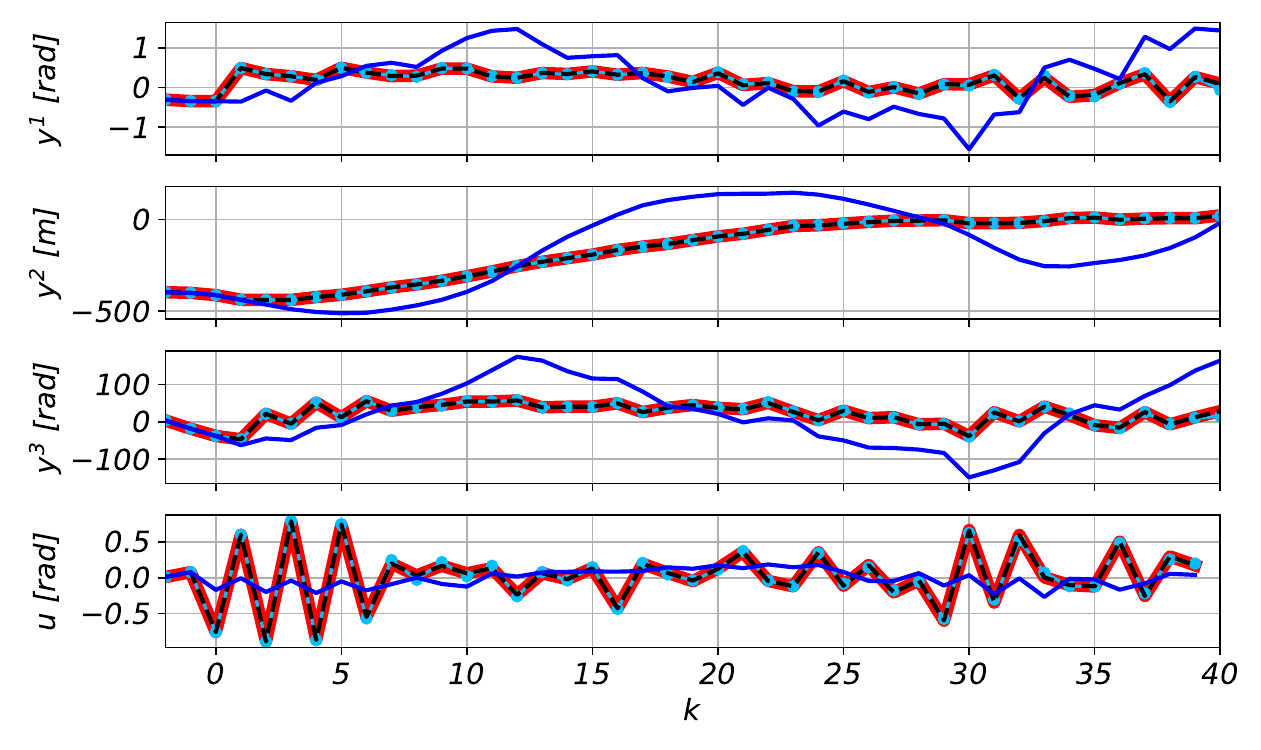}
		\caption{Aircraft example with Gaussian disturbance. Red-solid line: Scheme~\ref{alg:meas}; blue-solid line with circle marker: Scheme~\ref{alg:est}; black-dashed line: Scheme~\ref{alg:trun}; deep blue-solid line: Scheme~\ref{alg:deter}.} \label{fig:AircraftComparison} 		
	\end{center}
\end{figure}

\begin{table}[t!]
	\caption{Comparison of the computation times in \texttt{Julia} and the closed loop costs for the realized closed-loop trajectories in Figure~\ref{fig:AircraftComparison}}.
	\label{tab:AircraftComparisonTime}
	\centering
	\begin{adjustbox}{width=\columnwidth,center}
		\begin{tabular}{cccccc}
			\toprule
			\multirow{2}{*}{\shortstack{\\ \\ Data-driven\\ scheme}} &  \multicolumn{2}{c}{OCP} & \multicolumn{2}{c}{Disturbance estimation} & 	\multirow{2}{*}{$J^\text{cl}$ $[-]$}\\
			\cmidrule(lr){2-3} \cmidrule(lr){4-5} & Mean $\SI{}{[s]}$ & SD $\SI{}{[s]}$  & Mean $\SI{}{[s]}$ & SD $\SI{}{[s]}$ &\\
			\midrule
			\ref{alg:meas} & 0.788 & 0.110 & n.a. & n.a. & $1.339\times 10^{3}$ \\
			\ref{alg:est} & 0.795 & 0.115 & 0.021 & 0.011 & $1.220\times 10^{3}$\\
			\ref{alg:trun} & 0.399 & 0.245 & 0.019 & 0.011 & $1.274\times 10^{3}$\\
			\ref{alg:deter} & 0.194 & 0.052 & n.a. & n.a.  & $1.622\times 10^{5}$\\
			\bottomrule
		\end{tabular}
	\end{adjustbox}
\end{table}

\begin{figure}[t!]
	\begin{center}
		\includegraphics[width=8.4cm]{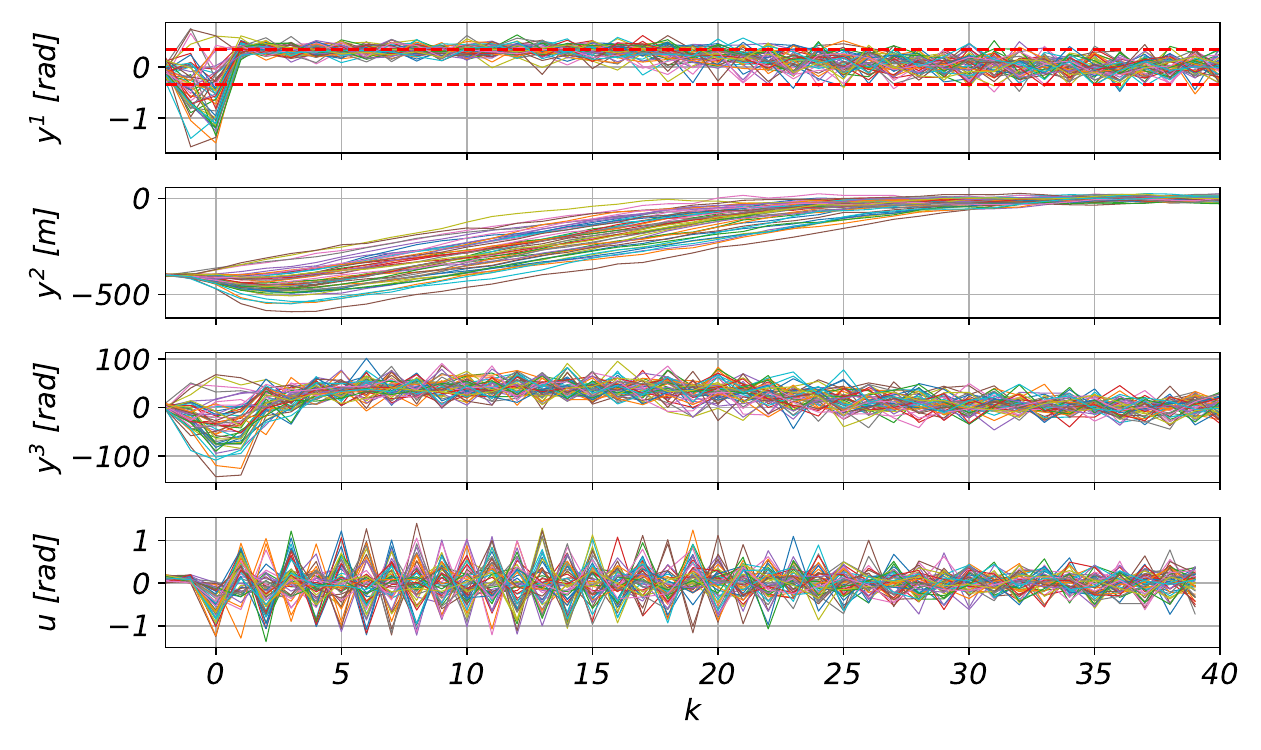}
		\caption{50 different closed-loop realization trajectories of Scheme~\ref{alg:est}. The red-dashed lines represent the chance constraints.} \label{fig:AircraftTraj} 		
	\end{center}
\end{figure}

\begin{figure}[t!]
	\begin{center}
		\includegraphics[width=0.3975\textwidth]{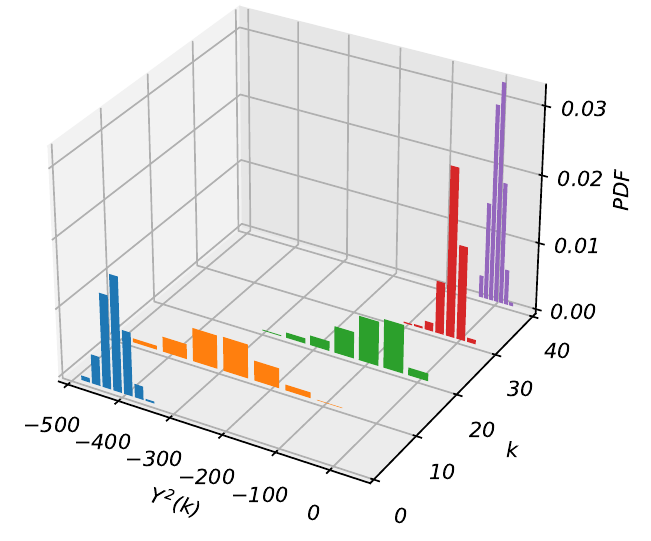}
		\caption{Histograms of the output $Y^2$ from closed-loop realization trajectories of Scheme~\ref{alg:est}.}
		\label{fig:AircraftDistEvolution}
	\end{center}
\end{figure}

%% file: Resource/Sec6_Conclusion.tex
\section{Conclusions and Outlook}\label{sec:conclusion}

This paper has addressed the extension of data-driven control and system analysis towards stochastic systems. Specifically, we have given an extension of the fundamental lemma for  stochastic LTI systems. The crucial insight of our analysis is that instead of formulating Hankel matrices in terms of random variables, it suffices to consider Hankel matrices constructed from past input-disturbance-output realizations. We have formalized this insight in terms of results on column-space equivalence, respectively, column-space inclusion and we have leveraged it to propose a framework for data-enabled uncertainty quantification and propagation via polynomial chaos expansions. Moreover, we have shown that Hankel matrices formulated directly in terms of random variables do not necessarily allow to characterize the full system behaviour.

As a by-product for our analysis, we have touched upon the estimation of past disturbance realizations from input-output data. Finally, we have shown by means of simulation examples that the proposed approach paves the road towards data-driven stochastic optimal and predictive control. 

At the same time, our results highlight the need for further and deeper investigations. This entails 1) a full-fledged behavioral characterization of stochastic systems, their PCE representations, and the relation between both,\footnote{In a follow-up to this paper we provide first results in this direction~\cite{tudo:faulwasser22f}.}  2) the consideration of output data corrupted by measurement noise, 3) the robustness analysis with respect to truncation errors in the series expansions and with respect to errors in the estimation of past disturbance realizations, and 4) the closed-loop analysis of data-driven stochastic output-feedback predictive control.

%% file: Resource/Sec7_Appendix.tex
\section*{Appendices}\label{sec:appen} 
\subsection{Galerkin Projection} 
For the sake of completeness, we recap the concept of Galerkin Projection, cf. \cite{GhanSpan03,muhlpfordt20}. Given the map $f: \splx{n_x} \to \splx{n_y}$ and the PCE of $X$ with respect to  the orthogonal polynomials $\{\phi^j\}_{j=0}^{\infty}$ Galerkin projection can be used to obtain the PCE coefficients of the image random variable $Y = f(X)$. 
It consists of the following steps:
\begin{itemize}
	\item[1)] Substitute X and Y with their PCEs
	\[ \sum_{j=0}^{\infty} \pce{y}^j \phi^j = f\left(\sum_{j=0}^{\infty} \pce{x}^j \phi^j\right).
	\]
	\item[2)] For all $i \in \N$, project onto the basis $\phi^i$ 
	\[
	\left\langle \sum_{j=0}^{\infty} \pce{y}^j \phi^j, \phi^i\right\rangle = \left\langle f\left(\sum_{j=0}^{\infty} \pce{x}^j \phi^j\right), \phi^i\right\rangle.
	\]
	\item[3)] Solve
	\[
	\pce{y}^i = \frac{\left\langle f\left(\sum_{j=0}^{\infty} \pce{x}^j \phi^j\right), \phi^i\right\rangle}{\langle\phi^i\rangle^2}. 
	\]
\end{itemize}
In case of an affine mapping $f(x) = Ax+b$, we have
\begin{align*}
	\pce{y}^i &= \frac{\left\langle \sum_{j=0}^{\infty} \left(A\pce{x}^j+b\right) \phi^j, \phi^i\right\rangle}{\langle\phi^i\rangle^2},\\
	&= \frac{\left\langle A\sum_{j=0}^{\infty} \pce{x}^j \phi^j + \sum_{j=0}^{\infty} b \phi^j, \phi^i\right\rangle}{\langle\phi^i\rangle^2} = A\pce{x}^i+b,
\end{align*}
which follows from the affinity of $f$ and  the orthogonality of $\{\phi^j\}_{j=0}^{\infty}$.

\subsection{Proof of Proposition \ref{pro:no_truncation_error}} 
\begin{proof}
For the sake of readability, we omit the subscript $[0,N-1]$ whenever there is no ambiguity. 

\textit{Part 1):}
With Assumption~\ref{ass:exact_Ini}, since \eqref{eq:bases} is the union of the independent bases $\phi_\text{{ini}}$ and $\phi_k$, $k \in \I_{[0,N-1]}$, we observe that $\Xini$ and $W_k$, $k \in \I_{[0,N-1]}$ also admit exact PCEs in the constructed finite-dimensional basis~\eqref{eq:finite_basis}. Moreover, for linear systems~\eqref{eq:OCPdynamic}-\eqref{eq:OCPdynamic_Y}, $\mbf{X}^\star$ and $\mbf{Y}^\star$ are linked to $\Xini$, $W_k$, $k \in I_{[0,N-1]}$, and $\mbf{U}^\star$ via an affine mapping. Hence, suppose that $\mbf{U}^\star$ admits an exact PCE with $L$ terms in the basis~\eqref{eq:finite_basis}, $\mbf{X}^\star$ and $\mbf{Y}^\star$ also admit exact PCEs with $L$ terms in the same basis as $\mbf{U}^\star$, cf. \cite{muehlpfordt18comments}.

\textit{Part 2):} The proof is done by contradiction.  That is,  consider the infinite orthogonal basis $ \{\phi^j\}_{j=0}^{\infty}$ whose first $L$ terms are given by \eqref{eq:finite_basis}. We suppose that $\mbf{U}^\star$ admits non-zero PCE coefficients beyond basis~\eqref{eq:finite_basis}, then we prove that there exists a feasible trajectory with a smaller value of the objective function.

For the sake of contradiction, suppose $\mbf{U}^\star$ has non-zero PCE coefficients beyond basis~\eqref{eq:finite_basis}. Then, this also applies to $\mbf{Y}^\star$ and $\mbf{X}^\star$, cf. the proof of statement 1). In other words, there exists a $\tilde{k}\in \I_{[0,N-1]}$ and $\tilde{j}\geq L$ such that
	\begin{equation}\label{eq:nonzeroPCE}
		\pce{u}_{\tilde{k}}^{\tilde{j},\star} \neq 0,~ \text{for some}~\tilde{j}\geq L. 
	\end{equation}
Consider the  truncation of $(\mbf{X}^\star,\mbf{Y}^\star, \mbf{U}^\star)$ after the first $L$ PCE terms via $\Pi^{\mbb{L}}$ with $\mbb{L} = \mbb{I}_{[0, L-1]}$ from \eqref{eq:PIprojection},  i.e. \[ \mathbf{(\bar{X},\bar{U},\bar{Y})} = 
	\Pi^{\mbb{L}}\left(\mbf{\left(X^\star,U^\star,Y^\star\right)}\right).\]
Since the trajectory tuple $(\pce{x}_k^{j,\star},\pce{y}_k^{j,\star},\pce{u}_k^{j,\star})$, $k\in\I_{[0,N-1]}$ satisfies the PCE dynamics \eqref{eq:PCEcoesDynamics} for every $j\in \I_{[0,L-1]}$, it is straightforward to see that $(\bar{X}_k,\bar{Y}_k,\bar{U}_k)$, $k\in\I_{[0,N-1]}$ satisfies~\eqref{eq:RVdynamics}. Moreover, in absence of chance constraints, $(\bar{\mbf{X}},\bar{\mbf{U}},\bar{\mbf{Y}})$ is a feasible trajectory of Problem~\ref{Problem1}. 

In addition, we reformulate the objective function \eqref{eq:stochasticOCP_obj} as
\[
\sum_{k=0}^{N-1} \sum_{j=0}^{\infty} \Big(\| \pcecoe{y}{j}_{k}\|^2_Q +\|\pcecoe{u}{j}_{k}\|^2_R\Big)\langle \phi^j\rangle^2.
\]
Observe that
\[ 
\sum_{k=0}^{N-1} \sum_{j=L}^{\infty}  \Big(\| \pcecoe{y}{j,\star}_{k}\|^2_Q +\|\pcecoe{u}{j,\star}_{k}\|^2_R\Big) \geq  \|\pcecoe{u}{\tilde{j},\star}_{\tilde{k}}\|^2_R >0
\] holds for $Q\succeq 0$,  $R \succ 0$ and all $\pcecoe{u}{\tilde{j},\star}_{\tilde{k}}$, $\tilde j \in \mbb{I}_{[L, \infty)}$, $\tilde{k}\in\I_{[0,N-1]}$ satisfying \eqref{eq:nonzeroPCE}.
Hence 
$(\bar{\mbf{X}},\bar{\mbf{U}},\bar{\mbf{Y}})$ admits a smaller objective value than $(\mbf{X}^\star,\mbf{U}^\star,\mbf{Y}^\star)$.  Thus we arrive at a contradiction which shows statement 2).

\textit{Part 3):} The proof of Part  2) shows that the truncated solution $\mathbf{(\bar{X},\bar{U},\bar{Y})} = 
\Pi^{\mbb{L}}\left(\mbf{\left(X^\star,U^\star,Y^\star\right)}\right)$ admits a smaller objective value than $(\mbf{X}^\star,\mbf{U}^\star,\mbf{Y}^\star)$. If $(\bar{\mbf{X}},\bar{\mbf{U}},\bar{\mbf{Y}})$ is a feasible solution, it will contradict the statement that $(\mbf{X}^\star,\mbf{U}^\star,\mbf{Y}^\star)$ is an optimal solution. Thus, if $\mbf{U}^\star$ does admit non-zero PCE coefficients beyond first $L$ terms, it implies that $(\bar{\mbf{X}},\bar{\mbf{U}},\bar{\mbf{Y}})$ is infeasible. In other words, the consideration of extra basis terms beyond~\eqref{eq:finite_basis} has to contribute to  feasibility. 
\end{proof}

\subsection{Proof of Theorem \ref{thm:equivOCP}} 
As a preparatory step, the following lemma gives a sufficient condition when Part 3) of Proposition~\ref{pro:no_truncation_error} can be strengthened without the a-priori assumption of $\mbf{U}^\star$ admitting an exact PCE.
\begin{lemma}[Exact PCE solution with box constraints]\label{lem:exactPCE_chance}
	Let Assumption~\ref{ass:exact_Ini} hold. Consider the optimal solution $\mbf{(X^\star,U^\star,Y^\star)}_{[0,N-1]}$ of Problem~\ref{Problem1} with  finite horizon $N \in \N^+$. If the chance constraint reformulation from \eqref{eq:chance_U}--\eqref{eq:chance_Y} to \eqref{eq:chance_reformulation} is exact and \eqref{eq:chance_U}--\eqref{eq:chance_Y} are box constraints,  i.e. $\mbb{Z} = [\underline z, \bar z]$ for $\mbb{Z} \in \{\mbb{U},\mbb{Y}\}$, then  $\mbf{(X^\star,U^\star,Y^\star)}_{[0,N-1]} $ admits an exact PCE with respect to the finite dimensional basis~\eqref{eq:finite_basis}. \End
\end{lemma}
\begin{proof}
	The proof follows the same ideas as the proof in Part~2) of Proposition \ref{pro:no_truncation_error}. The crucial difference in our proof below is that the feasibility of the truncated solution $(\bar{\mbf{X}},\bar{\mbf{U}},\bar{\mbf{Y}})$ holds for the given assumptions. 
	
	First, with non-zero PCE coffcients~\eqref{eq:nonzeroPCE}, from \eqref{eq:MomentsPCE} we note that $\mbb V[\bar{Z}_k]\leq \mbb V [Z_k^\star]$, $Z\in\{U,Y\}$. 
	Therefore, if $(\mbf{U}^\star,\mbf{Y}^\star)$ satisfy the chance constraint reformulation~\eqref{eq:chance_reformulation} with $\mbb{Z} = [\underline z, \bar z]$ for $\mbb{Z} \in \{\mbb{U},\mbb{Y}\}$, the truncated variables $(\bar{\mbf{X}},\bar{\mbf{U}},\bar{\mbf{Y}})$ also satisfy \eqref{eq:chance_reformulation} with $\mbb{Z} = [\underline z, \bar z]$ for $\mbb{Z} \in \{\mbb{U},\mbb{Y}\}$.
Hence, we conclude the feasibility of $(\bar{\mbf{X}},\bar{\mbf{U}},\bar{\mbf{Y}})$ whenever the reformulation between \eqref{eq:chance_U}--\eqref{eq:chance_Y} to \eqref{eq:chance_reformulation}  is exact and when \eqref{eq:chance_U}--\eqref{eq:chance_Y} are box constraints.
\end{proof}

Now, we are ready to prove Theorem~\ref{thm:equivOCP}.

\begin{proof}
	\textit{Part 1)}: 
First, we prove that with Assumption~\ref{ass:exact_Ini_UY}, there exists a $\Xini$ that admits an exact PCE with $L_\text{ini}$ terms.
 With Assumption~\ref{ass:exact_Ini_UY}, we have
 $Z_k = \sum_{j=0}^{L_{\text{ini}}-1}\pce{z}_k^j\phi_{\text{ini}}^j$ for $k \in \I_{[-\Tini,-1]}$ and $(Z,\pce{z}) \in \{(U,\pce{u}),(W,\pce{w}),(Y,\pce{y})\} $.
 For each $j \in \I_{[0, L_{\text{ini}}]}$, the PCE coefficient trajectory  $(\pce{u}^j,\pce{w}^j,\pce{y}^j)_{[-\Tini,-1]}$ satisfies~\eqref{eq:PCEcoesDynamics}. Since \eqref{eq:PCEcoesDynamics} is observable, with $\Tini$ being larger than the system lag, we can determine an internal state $\pce{x}^j_\text{ini}$ for the PCE coefficients trajectory at time instant $0$ from $(\pce{u}^j,\pce{w}^j,\pce{y}^j)_{[-\Tini,-1]}$. Thus, one can recover the internal state  $\Xini=\sum_{j=0}^{L_{\text{ini}}-1} \pce{x}^j_\text{ini}\phi_\text{ini}^j$ for the random variable trajectory and this $\Xini$ admits an exact PCE with $L_\text{ini}$ terms.

Second, the data-driven reformulation of \eqref{eq:PCEdynamicinOCP} and \eqref{eq:PCEdynamicinOCP_y} to \eqref{eq:H_PCE_SOCP_hankel} follows Corollary~ \ref{cor:mixed_funda} with $\trar{(u,w)}{T-1}$ persistently exciting of order $N+\Tini+\dimx$.  Thus, we conclude that  with $\Tini$ being larger than the system lag, one can recover the internal state $\Xini$ from the consistency data in Assumption~\ref{ass:exact_Ini_UY}, and with this $\Xini$ Problem~\ref{Problem2} admits the same optimal input-output solution set as Problem~\ref{Problem3}, i.e. $\mcl S_2 = \mcl S_3$.

\textit{Part 2)}:  Suppose the chance constraint reformulation from \eqref{eq:chance_U}--\eqref{eq:chance_Y} to \eqref{eq:chance_reformulation} is equivalent (exact) and  \eqref{eq:chance_U}--\eqref{eq:chance_Y} are box constraints.  As shown in Lemma~\ref{lem:exactPCE_chance}, for any $\Xini$ with an exact PCE of $L_\text{ini}$ terms,   $\mbf{(X^\star,U^\star,Y^\star)}_{[0,N-1]} $ admit exact PCEs with respect to the finite-dimensional basis~\eqref{eq:finite_basis}.
 Thus, the PCE reformulation of random variables from Problem~\ref{Problem1} to Problem~\ref{Problem2} is without truncation error. Hence, we conclude that $\mcl S_1 = \Psi(\mcl S_2)$ holds whereby $\Psi(\mcl S_2)$ is the element-wise application of the map.
\end{proof}

%% file: StochFundamentalLemmaMPC.bbl
\begin{thebibliography}{10}
\providecommand{\url}[1]{#1}
\csname url@samestyle\endcsname
\providecommand{\newblock}{\relax}
\providecommand{\bibinfo}[2]{#2}
\providecommand{\BIBentrySTDinterwordspacing}{\spaceskip=0pt\relax}
\providecommand{\BIBentryALTinterwordstretchfactor}{4}
\providecommand{\BIBentryALTinterwordspacing}{\spaceskip=\fontdimen2\font plus
\BIBentryALTinterwordstretchfactor\fontdimen3\font minus
  \fontdimen4\font\relax}
\providecommand{\BIBforeignlanguage}[2]{{%
\expandafter\ifx\csname l@#1\endcsname\relax
\typeout{** WARNING: IEEEtran.bst: No hyphenation pattern has been}%
\typeout{** loaded for the language `#1'. Using the pattern for}%
\typeout{** the default language instead.}%
\else
\language=\csname l@#1\endcsname
\fi
#2}}
\providecommand{\BIBdecl}{\relax}
\BIBdecl

\bibitem{Willems2005}
J.~C. Willems, P.~Rapisarda, I.~Markovsky, and B.~L.~M. De~Moor, ``A note on
  persistency of excitation,'' \emph{Systems \& Control Letters}, vol.~54,
  no.~4, pp. 325--329, 2005.

\bibitem{Alsalti21}
M.~Alsalti, J.~Berberich, V.~G. Lopez, F.~Allg{\"o}wer, and M.~A. M{\"u}ller,
  ``Data-based system analysis and control of flat nonlinear systems,'' in
  \emph{2021 60th IEEE Conference on Decision and Control (CDC)}.\hskip 1em
  plus 0.5em minus 0.4em\relax IEEE, 2021, pp. 1484--1489.

\bibitem{lian21k}
Y.~Lian, R.~Wang, and C.~N. Jones, ``Koopman based data-driven predictive
  control,'' \emph{arXiv preprint arXiv:2102.05122}, 2021.

\bibitem{Verhoek21}
C.~Verhoek, H.~S. Abbas, R.~T{\'o}th, and S.~Haesaert, ``Data-driven predictive
  control for linear parameter-varying systems,'' \emph{IFAC-PapersOnLine},
  vol.~54, no.~8, pp. 101--108, 2021, 4th IFAC Workshop on Linear Parameter
  Varying Systems LPVS 2021.

\bibitem{Allibhoy20}
A.~Allibhoy and J.~Cort{\'e}s, ``Data-based receding horizon control of linear
  network systems,'' \emph{IEEE Control Systems Letters}, vol.~5, no.~4, pp.
  1207--1212, 2020.

\bibitem{Mishra20}
V.~K. Mishra, I.~Markovsky, and B.~Grossmann, ``Data-driven tests for
  controllability,'' \emph{IEEE Control Systems Letters}, vol.~5, no.~2, pp.
  517--522, 2020.

\bibitem{Yu2021}
Y.~Yu, S.~Talebi, H.~J. van Waarde, U.~Topcu, M.~Mesbahi, and
  B.~A{\c{c}}{\i}kmeșe, ``{On controllability and persistency of excitation in
  data-driven control: Extensions of Willems’ fundamental lemma},'' in
  \emph{2021 60th IEEE Conference on Decision and Control (CDC)}.\hskip 1em
  plus 0.5em minus 0.4em\relax IEEE, 2021, pp. 6485--6490.

\bibitem{Berberich21}
J.~Berberich, J.~K{\"o}hler, M.~A. M{\"u}ller, and F.~Allg{\"o}wer, ``{Linear
  tracking MPC for nonlinear systems—Part II: the data-driven case},''
  \emph{IEEE Transactions on Automatic Control}, vol.~67, no.~9, pp.
  4406--4421, 2022.

\bibitem{Martinelli2022}
A.~Martinelli, M.~Gargiani, M.~Draskovic, and J.~Lygeros, ``{Data-driven
  optimal control of affine systems: A linear programming perspective},''
  \emph{IEEE Control Systems Letters}, 2022.

\bibitem{DePersis19}
C.~De~Persis and P.~Tesi, ``Formulas for data-driven control: Stabilization,
  optimality, and robustness,'' \emph{IEEE Transactions on Automatic Control},
  vol.~65, no.~3, pp. 909--924, 2019.

\bibitem{Markovsky21r}
I.~Markovsky and F.~D{\"o}rfler, ``Behavioral systems theory in data-driven
  analysis, signal processing, and control,'' \emph{Annual Reviews in Control},
  vol.~52, pp. 42--64, 2021.

\bibitem{VanWaarde20}
H.~J. van Waarde, J.~Eising, H.~L. Trentelman, and M.~K. Camlibel, ``Data
  informativity: A new perspective on data-driven analysis and control,''
  \emph{IEEE Transactions on Automatic Control}, vol.~65, no.~11, pp.
  4753--4768, 2020.

\bibitem{Coulson2019}
J.~Coulson, J.~Lygeros, and F.~D{\"o}rfler, ``Data-enabled predictive control:
  In the shallows of the {DeePC},'' in \emph{2019 18th European Control
  Conference (ECC)}.\hskip 1em plus 0.5em minus 0.4em\relax IEEE, 2019, pp.
  307--312.

\bibitem{Yang15}
H.~Yang and S.~Li, ``A data-driven predictive controller design based on
  reduced {Hankel} matrix,'' in \emph{2015 10th Asian Control Conference
  (ASCC)}.\hskip 1em plus 0.5em minus 0.4em\relax IEEE, 2015, pp. 1--7.

\bibitem{Berberich20}
J.~Berberich, J.~K{\"o}hler, M.~A. M{\"u}ller, and F.~Allg{\"o}wer,
  ``Data-driven model predictive control with stability and robustness
  guarantees,'' \emph{IEEE Transactions on Automatic Control}, vol.~66, no.~4,
  pp. 1702--1717, 2020.

\bibitem{Huang19}
L.~Huang, J.~Coulson, J.~Lygeros, and F.~D{\"o}rfler, ``Data-enabled predictive
  control for grid-connected power converters,'' in \emph{2019 IEEE 58th
  Conference on Decision and Control (CDC)}.\hskip 1em plus 0.5em minus
  0.4em\relax IEEE, 2019, pp. 8130--8135.

\bibitem{Lian21}
Y.~Lian, J.~Shi, M.~P. Koch, and C.~N. Jones, ``Adaptive robust data-driven
  building control via bi-level reformulation: An experimental result,''
  \emph{arXiv preprint arXiv:2106.05740}, 2021.

\bibitem{Berberich21at}
J.~Berberich, J.~K{\"o}hler, M.~A. M{\"u}ller, and F.~Allg{\"o}wer,
  ``Data-driven model predictive control: Closed-loop guarantees and
  experimental results,'' \emph{at-Automatisierungstechnik}, vol.~69, no.~7,
  pp. 608--618, 2021.

\bibitem{Carlet20}
P.~G. Carlet, A.~Favato, S.~Bolognani, and F.~D{\"o}rfler, ``Data-driven
  predictive current control for synchronous motor drives,'' in \emph{2020 IEEE
  Energy Conversion Congress and Exposition (ECCE)}.\hskip 1em plus 0.5em minus
  0.4em\relax IEEE, 2020, pp. 5148--5154.

\bibitem{Bilgic22}
D.~Bilgic, A.~Koch, G.~Pan, and T.~Faulwasser, ``Toward data-driven predictive
  control of multi-energy distribution systems,'' \emph{Electric Power Systems
  Research}, vol. 212, p. 108311, 2022.

\bibitem{Berberich20r}
J.~Berberich, A.~Koch, C.~W. Scherer, and F.~Allg{\"o}wer, ``Robust data-driven
  state-feedback design,'' in \emph{2020 American Control Conference
  (ACC)}.\hskip 1em plus 0.5em minus 0.4em\relax IEEE, 2020, pp. 1532--1538.

\bibitem{De21l}
C.~De~Persis and P.~Tesi, ``Low-complexity learning of linear quadratic
  regulators from noisy data,'' \emph{Automatica}, vol. 128, p. 109548, 2021.

\bibitem{Van20n}
H.~J. van Waarde, M.~K. Camlibel, and M.~Mesbahi, ``From noisy data to feedback
  controllers: Non-conservative design via a matrix {S}-lemma,'' \emph{IEEE
  Transactions on Automatic Control}, 2020.

\bibitem{Yin20}
M.~Yin, A.~Iannelli, and R.~S. Smith, ``Maximum likelihood estimation in
  data-driven modeling and control,'' \emph{IEEE Transactions on Automatic
  Control}, vol.~68, no.~1, pp. 317--328, 2023.

\bibitem{Coulson21a}
J.~Coulson, J.~Lygeros, and F.~Dörfler, ``Distributionally robust chance
  constrained data-enabled predictive control,'' \emph{IEEE Transactions on
  Automatic Control}, vol.~67, no.~7, pp. 3289--3304, 2022.

\bibitem{Willems12}
J.~C. Willems, ``Open stochastic systems,'' \emph{IEEE Transactions on
  Automatic Control}, vol.~58, no.~2, pp. 406--421, 2012.

\bibitem{Baggio17}
G.~Baggio and R.~Sepulchre, ``{LTI stochastic processes: a behavioral
  perspective},'' \emph{IFAC-PapersOnLine}, vol.~50, no.~1, pp. 2806--2811,
  2017.

\bibitem{tudo:faulwasser22f}
T.~Faulwasser, R.~Ou, G.~Pan, P.~Schmitz, and K.~Worthmann, ``{Behavioral
  theory for stochastic systems? A data-driven journey from Willems to Wiener
  and back again},'' \emph{Arxiv Preprint Arxiv:2209.06414}, 2022.

\bibitem{sullivan15introduction}
T.~J. Sullivan, \emph{{Introduction to Uncertainty Quantification}}.\hskip 1em
  plus 0.5em minus 0.4em\relax Springer, 2015, vol.~63.

\bibitem{Mesbah16}
A.~Mesbah, ``Stochastic model predictive control: An overview and perspectives
  for future research,'' \emph{IEEE Control Systems Magazine}, vol.~36, no.~6,
  pp. 30--44, 2016.

\bibitem{Heirung18}
T.~A.~N. Heirung, J.~A. Paulson, J.~O’Leary, and A.~Mesbah, ``Stochastic
  model predictive control—how does it work?'' \emph{Computers \& Chemical
  Engineering}, vol. 114, pp. 158--170, 2018.

\bibitem{Fagiano12}
L.~Fagiano and M.~Khammash, ``Nonlinear stochastic model predictive control via
  regularized polynomial chaos expansions,'' in \emph{2012 IEEE 51st IEEE
  Conference on Decision and Control (CDC)}.\hskip 1em plus 0.5em minus
  0.4em\relax IEEE, 2012, pp. 142--147.

\bibitem{kim13wiener}
K.~K. Kim, D.~E. Shen, Z.~K. Nagy, and R.~D. Braatz, ``Wiener's polynomial
  chaos for the analysis and control of nonlinear dynamical systems with
  probabilistic uncertainties [{H}istorical {P}erspectives],'' \emph{IEEE
  Control Systems Magazine}, vol.~33, no.~5, pp. 58--67, 2013.

\bibitem{paulson14fast}
J.~A. Paulson, A.~Mesbah, S.~Streif, R.~Findeisen, and R.~D. Braatz, ``Fast
  stochastic model predictive control of high-dimensional systems,'' in
  \emph{2014 53rd IEEE Conference on Decision and Control (CDC)}.\hskip 1em
  plus 0.5em minus 0.4em\relax IEEE, 2014, pp. 2802--2809.

\bibitem{Mesbah14}
A.~Mesbah, S.~Streif, R.~Findeisen, and R.~D. Braatz, ``Stochastic nonlinear
  model predictive control with probabilistic constraints,'' in \emph{2014
  American Control Conference (ACC)}.\hskip 1em plus 0.5em minus 0.4em\relax
  IEEE, 2014, pp. 2413--2419.

\bibitem{Sergio17}
S.~Lucia, J.~A. Paulson, R.~Findeisen, and R.~D. Braatz, ``On stability of
  stochastic linear systems via polynomial chaos expansions,'' in \emph{2017
  American Control Conference (ACC)}, 2017, pp. 5089--5094.

\bibitem{Lefebvre20}
T.~Lefebvre, ``On moment estimation from polynomial chaos expansion models,''
  \emph{IEEE Control Systems Letters}, vol.~5, no.~5, pp. 1519--1524, 2020.

\bibitem{Ahbe20}
E.~Ahbe, A.~Iannelli, and R.~S. Smith, ``Region of attraction analysis of
  nonlinear stochastic systems using {Polynomial Chaos Expansion},''
  \emph{Automatica}, vol. 122, p. 109187, 2020.

\bibitem{Till19}
T.~Mühlpfordt, L.~Roald, V.~Hagenmeyer, T.~Faulwasser, and S.~Misra,
  ``Chance-constrained {AC} optimal power flow: A polynomial chaos approach,''
  \emph{IEEE Transactions on Power Systems}, vol.~34, no.~6, pp. 4806--4816,
  2019.

\bibitem{fristedt13modern}
B.~E. Fristedt and L.~F. Gray, \emph{A Modern Approach to Probability
  Theory}.\hskip 1em plus 0.5em minus 0.4em\relax Springer Science \& Business
  Media, 2013.

\bibitem{wiener38homogeneous}
N.~Wiener, ``The homogeneous chaos,'' \emph{American Journal of Mathematics},
  pp. 897--936, 1938.

\bibitem{field04accuracy}
R.~V. Field~Jr. and M.~Grigoriu, ``On the accuracy of the polynomial chaos
  approximation,'' \emph{Probabilistic Engineering Mechanics}, vol.~19, no.~1,
  pp. 65--80, 2004, fourth International Conference on Computational Stochastic
  Mechanics.

\bibitem{muehlpfordt18comments}
T.~M{\"u}hlpfordt, R.~Findeisen, V.~Hagenmeyer, and T.~Faulwasser, ``Comments
  on quantifying truncation errors for polynomial chaos expansions,''
  \emph{IEEE Control Systems Letters}, vol.~2, no.~1, pp. 169--174, 2018.

\bibitem{koekoek96askey}
R.~Koekoek and R.~F. Swarttouw, ``The {Askey-scheme} of hypergeometric
  orthogonal polynomials and its q-analogue,'' Department of Technical
  Mathematics and Informatics, Delft University of Technology, Delft, The
  Netherlands, Tech. Rep. 98-17, 1998.

\bibitem{xiu02wiener}
D.~Xiu and G.~E. Karniadakis, ``The {W}iener--{A}skey polynomial chaos for
  stochastic differential equations,'' \emph{SIAM Journal on Scientific
  Computing}, vol.~24, no.~2, pp. 619--644, 2002.

\bibitem{Witteveen2006}
J.~A.~S. Witteveen and H.~Bijl, ``Modeling arbitrary uncertainties using
  {Gram-Schmidt} polynomial chaos,'' in \emph{44th AIAA Aerospace Sciences
  Meeting and Exhibit}, N.~J. Pfeiffer, Ed.\hskip 1em plus 0.5em minus
  0.4em\relax United States: American Institute of Aeronautics and Astronautics
  Inc. (AIAA), 2006, p. 896.

\bibitem{farina13probabilistic}
M.~Farina, L.~Giulioni, L.~Magni, and R.~Scattolini, ``A probabilistic approach
  to model predictive control,'' in \emph{2013 52nd IEEE Conference on Decision
  and Control (CDC)}.\hskip 1em plus 0.5em minus 0.4em\relax IEEE, 2013, pp.
  7734--7739.

\bibitem{Weissel2009}
F.~Weissel, M.~F. Huber, and U.~D. Hanebeck, ``Stochastic nonlinear model
  predictive control based on {Gaussian} mixture approximations,'' in
  \emph{Informatics in Control, Automation and Robotics}.\hskip 1em plus 0.5em
  minus 0.4em\relax Springer, 2009, pp. 239--252.

\bibitem{Schildbach2014}
G.~Schildbach, L.~Fagiano, C.~Frei, and M.~Morari, ``The scenario approach for
  stochastic model predictive control with bounds on closed-loop constraint
  violations,'' \emph{Automatica}, vol.~50, no.~12, pp. 3009--3018, 2014.

\bibitem{markovsky06}
I.~Markovsky, J.~C. Willems, S.~Van~Huffel, and B.~De~Moor, \emph{{Exact and
  approximate modeling of linear systems: A behavioral approach}}.\hskip 1em
  plus 0.5em minus 0.4em\relax SIAM, 2006.

\bibitem{o2021data}
E.~O’Dwyer, E.~C. Kerrigan, P.~Falugi, M.~Zagorowska, and N.~Shah,
  ``Data-driven predictive control with improved performance using segmented
  trajectories,'' \emph{IEEE Transactions on Control Systems Technology}, pp.
  1--11, 2022.

\bibitem{Bock1984}
H.~G. Bock and K.~J. Plitt, ``A multiple shooting algorithm for direct solution
  of optimal control problems,'' \emph{IFAC Proceedings Volumes}, vol.~17,
  no.~2, pp. 1603--1608, 1984.

\bibitem{Ou23}
R.~Ou, G.~Pan, and T.~Faulwasser, ``Data-driven multiple shooting for
  stochastic optimal control,'' \emph{IEEE Control Systems Letters}, vol.~7,
  pp. 313--318, 2023.

\bibitem{Sadamoto2022}
T.~Sadamoto, ``On equivalence of data informativity for identification and
  data-driven control of partially observable systems,'' \emph{IEEE
  Transactions on Automatic Control}, pp. 1--8, 2022.

\bibitem{muehlpfordt19polychaos}
T.~M{\"u}hlpfordt, F.~Zahn, V.~Hagenmeyer, and T.~Faulwasser,
  ``{P}oly{C}haos.jl—{A Julia} package for polynomial chaos in systems and
  control,'' \emph{IFAC-PapersOnLine}, vol.~53, no.~2, pp. 7210--7216, 2020,
  21st IFAC World Congress.

\bibitem{Dunning17}
I.~Dunning, J.~Huchette, and M.~Lubin, ``{JuMP: A modeling language for
  mathematical optimization},'' \emph{SIAM Review}, vol.~59, no.~2, pp.
  295--320, 2017.

\bibitem{grune13economic}
L.~Gr{\"u}ne, ``Economic receding horizon control without terminal
  constraints,'' \emph{Automatica}, vol.~49, no.~3, pp. 725--734, 2013.

\bibitem{Ou21}
R.~Ou, M.~H. Baumann, L.~Gr{\"u}ne, and T.~Faulwasser, ``A simulation study on
  turnpikes in stochastic {LQ} optimal control,'' \emph{IFAC-PapersOnLine},
  vol.~54, no.~3, pp. 516--521, 2021, 16th IFAC Symposium on Advanced Control
  of Chemical Processes (ADCHEM).

\bibitem{maciejowski02predictive}
J.~M. Maciejowski, \emph{Predictive Control with Constraints}.\hskip 1em plus
  0.5em minus 0.4em\relax Pearson Education, 2002.

\bibitem{GhanSpan03}
R.~G. Ghanem and P.~D. Spanos, \emph{{Stochastic Finite Elements: A Spectral
  Approach}}, revised~ed.\hskip 1em plus 0.5em minus 0.4em\relax Springer New
  York, 2003.

\bibitem{muhlpfordt20}
T.~M{\"u}hlpfordt, ``Uncertainty quantification via polynomial chaos expansion
  -- {M}ethods and applications for optimization of power systems,'' Ph.D.
  dissertation, Karlsruher Institut f{\"u}r Technologie (KIT), 2020.

\end{thebibliography}
